\documentclass[a4paper,USenglish]{lipics-v2021}
\pdfoutput=1
\hideLIPIcs
\nolinenumbers

\makeatletter
\renewcommand\@oddfoot{}
\makeatother

\RequirePackage[T1]{fontenc}
\RequirePackage{lmodern}

\sffamily

\usepackage{helvet}
\sffamily
\DeclareFontShape{T1}{lmss}{bx}{sc} { <-> ssub * phv/bx/sc }{}
\DeclareFontShape{T1}{lmss}{m}{sc} { <-> ssub * phv/m/sc }{}

\makeatletter
\g@addto@macro\bfseries{\boldmath}
\makeatother

\RequirePackage[utf8]{inputenc}
\usepackage{microtype}

\usepackage{mathtools}

\usepackage[noabbrev,capitalize]{cleveref}

\theoremstyle{plain}
\usepackage{booktabs,makecell,multirow,array}
\newcolumntype{H}{>{\setbox0=\hbox\bgroup}c<{\egroup}@{}}

\usepackage{subcaption}
\usepackage{xspace}

\usepackage{tikz}
\usetikzlibrary{decorations.pathreplacing}
\usetikzlibrary{decorations.markings,arrows,arrows.meta}
\tikzset{
    marrow/.style={decoration={markings,mark=at position 0.65 
    with 
    {\arrow{#1}}}, postaction=decorate}
}

\usepackage{pgfplots}

\usepackage{algorithm}
\usepackage{algcompatible}
\usepackage[noend]{algpseudocode}

\algrenewcommand\algorithmicindent{1.6em}
\usepackage{xparse}
\makeatletter
\NewDocumentCommand{\LeftComment}{s m}{%
  \IfBooleanF{#1}{\hspace*{\ALG@thistlm}}\(\triangleright\) #2}
\makeatother

\renewcommand{\setminus}{\smallsetminus}

\DeclareMathOperator{\argmin}{arg\,min}

\newcommand{\In}{\textbf{in}\xspace}
\newcommand{\Next}{\textbf{next}\xspace}
\newcommand{\Continue}{\textbf{continue}\xspace}

\newcommand{\Pop}{\textrm{pop}}

\renewcommand{\And}{\textbf{and}\xspace}

\allowdisplaybreaks

\newcommand{\R}{\ensuremath{\text{\textbf{R}}}}

\newcommand{\eps}{\ensuremath{\varepsilon}}

\newcommand{\winzig}{\textnormal{tiny}}
\newcommand{\klein}{\textnormal{small}}
\newcommand{\mittel}{\textnormal{medium}}
\newcommand{\gross}{\textnormal{big}}
\newcommand{\riesig}{\textnormal{huge}}

\newcommand{\kleinlich}{\textnormal{little}}
\newcommand{\groesslich}{\textnormal{large}}

\newcommand{\paltry}{\textnormal{paltry}}
\newcommand{\precious}{\textnormal{precious}}
\newcommand{\spread}{\textnormal{spread}}
\newcommand{\round}{\textnormal{round}}

\newcommand{\limit}{\textnormal{limit}}
\renewcommand{\split}{\textnormal{split}}

\newcommand{\actual}{\textnormal{actual}}

\newcommand{\purelypaltry}{\textnormal{small}}
\newcommand{\levelup}{\textnormal{levelUp}}

\newcommand{\maxpaltry}{\textnormal{valueLimit}}
\newcommand{\temp}{\textnormal{temp}}

\newcommand{\matches}{\textnormal{matches}}
\newcommand{\lost}{\textnormal{lost}}
\newcommand{\seen}{\textnormal{seen}}

\newcommand{\backup}{\textnormal{copy}}

\let\leftold\left
\let\rightold\right
\renewcommand{\left}{\mathopen{}\mathclose\bgroup\leftold}
\renewcommand{\right}{\aftergroup\egroup\rightold}

\newcommand{\inst}{\ensuremath{I}}
\newcommand{\instset}{\ensuremath{\mathcal{I}}}
\newcommand{\size}{\ensuremath{s}}
\newcommand{\val}{\ensuremath{v}}

\newcommand{\appsol}{\ensuremath{T}}

\newcommand{\appalg}{\ensuremath{\mathcal{A}}}
\newcommand{\opt}{\ensuremath{\textnormal{\texttt{opt}}}}
\newcommand{\app}{\ensuremath{\textnormal{\texttt{alg}}}}

\newcommand{\Knap}{\textnormal{\textsc{Knap}}}

\newcommand{\PropKnap}{\textnormal{\textsc{Prop\-Knap}}\xspace}
\newcommand{\GenRemKnap}{\textnormal{\textsc{Rem\-Knap}}\xspace}
\newcommand{\PropRemKnap}{\textnormal{\textsc{Prop\-Rem\-Knap}}\xspace}

\title{Removable Online Knapsack and Advice}
\titlerunning{Removable Online Knapsack and Advice}

\author{Hans-Joachim Böckenhauer}{Department of Computer Science, ETH Zürich}{hjb@inf.ethz.ch}{https://orcid.org/0000-0001-9164-3674}{}
\author{Fabian Frei}{CISPA Helmholtz Center for Information Security}{fabian.frei@cispa.de}{https://orcid.org/0000-0002-1368-3205}{Work done in part while at ETH Zürich.}
\author{Peter Rossmanith}{Department of Computer Science, RWTH Aachen}{rossmani@cs.rwth-aachen.de}{https://orcid.org/0000-0003-0177-8028}{}


\authorrunning{H.-J.\ Böckenhauer, F.\ Frei, P.\ Rossmanith}

\Copyright{Hans-Joachim Böckenhauer, Fabian Frei, Peter Rossmanith}


\begin{document}
\maketitle
\begin{abstract}
In the \emph{proportional knapsack} problem, we are 
given a 
knapsack of some capacity and a set of variably 
sized items. 
The goal is to pack a selection of these items that fills the knapsack as much as possible. 
The \emph{online} version of this problem reveals the items and their sizes not all at once but one by one. 
For each item, the algorithm has to decide immediately whether to pack it or not. 
We consider a natural variant of this online 
knapsack problem, 
which has been coined \emph{removable knapsack}.
It differs from the classical variant by allowing the removal of any packed item from the knapsack. 
Repacking is impossible, however: Once an item is removed, it is gone for good. 
We analyze the \emph{advice complexity} of this problem. 
It measures how many \emph{advice bits} an omniscient oracle needs to provide for an online algorithm to reach any given \emph{competitive ratio}, which is---understood in its strict sense---just the algorithm's approximation factor. 
The online knapsack problem is known for its peculiar advice behavior involving three jumps in competitivity.
We show that the advice complexity of the 
version with removability 
is quite different but just as interesting: 
The competitivity starts from the golden ratio when no advice is given. 
It then drops down to $1+\eps$ 
for a constant 
amount of advice already, which requires 
logarithmic advice in the 
classical version. 
Removability comes as no relief to the 
perfectionist, however: 
Optimality still requires linear advice as 
before.
These results are particularly noteworthy from a 
structural viewpoint for the exceptionally slow transition 
from near-optimality to optimality. 

Our most important and demanding result shows 
that the \emph{general} knapsack 
problem, which allows 
an item's value to differ from its size, exhibits 
a similar behavior 
for removability, but with an even more 
pronounced jump from an 
unbounded competitive ratio to near-optimality 
within just 
constantly many advice bits. This is a unique behavior among the 
problems considered in the literature so far. 

An advice analysis is interesting in its own right, as it 
allows us to measure the information content of a problem and leads 
to structural insights. But it also provides insurmountable 
lower bounds, applicable to any kind of additional 
information about the instances, including predictions provided by 
machine-learning algorithms and artificial intelligence. 
Unexpectedly, advice algorithms are useful in various 
real-life situations, too. For example, they provide smart 
strategies for cooperation in winner-take-all competitions, where 
several participants pool together to implement different strategies 
and share the obtained prize. 
Further illustrating the versatility of our advice-complexity bounds, 
our results automatically improve some of the best
known lower bounds on the competitive ratio 
for removable knapsack with randomization. 
The presented advice algorithms also automatically yield 
deterministic 
algorithms 
for established deterministic models such as knapsack with a resource 
buffer 
and various problems with more than one knapsack. 
In their seminal paper introducing removability to the knapsack 
problem, Iwama and Taketomi have indeed proposed a multiple knapsack 
problem for which we can establish 
a one-to-one correspondence with the advice model; 
this paper therefore even provides a comprehensive analysis 
for this up until now neglected problem. 
\end{abstract}
\newpage
\section{Introduction}
In this first section, we briefly summarize what online algorithms 
and advice are, 
then informally present the problem whose advice complexity we will be analyzing, 
and finally describe several applications of such advice complexity results. 

\subsection{Online Algorithms and Advice Complexity}
\emph{Online} algorithms receive their input piece by piece and have to determine parts of the solution before knowing the entire instance. 
This often leaves them unable to compete with offline algorithms, which know the entire input in advance, in a meaningful way. 
In the \emph{advice} model, 
we assume an omniscient oracle that provides the online algorithm with some information on how to solve the upcoming instance best. 
If the oracle can communicate to the algorithm an unlimited amount of such advice, it will of course be able to lead the algorithm to an optimal solution for every instance. 
The \emph{advice complexity} measures the minimum amount of information necessary for the online algorithm to achieve any given approximation ratio, which is commonly called strict \emph{competitive ratio} or \emph{competitivity} in this context. 

Advice complexity is a well-established tool to gauge the
information content of an online  problem~\cite{BKKKM2009,EFKR2011,HKK2010}. 
For a detailed and careful introduction to the theory, we refer to the textbook by Komm~\cite{Kom2016}. 
Another classical textbook on online problems is 
written by Borodin 
and Yaniv~\cite{BE1998}.
The trade-off between a low number of transmitted advice bits on 
the one hand and achieving a good competitive ratio on the other 
hand has been examined for a wealth of problems---see the survey by 
Boyar et al.~\cite{BFKLM12017}---but one stands out for its 
peculiar behavior: the knapsack problem. 

\subsection{Knapsack and Removability}
A knapsack instance presents the online algorithm with a sequence of items of different sizes. 
Upon the arrival of each item, the algorithm has to decide whether to pack it into a knapsack or discard it. 
The goal is to fill the knapsack as much as 
possible without ever 
exceeding the knapsack's given capacity. 
This problem is sometimes also referred to as 
the 
\emph{proportional} or \emph{simple} knapsack 
problem, as opposed 
to the \emph{general} knapsack problem, in which 
every item has not 
only a size but also a value.\footnote{It is also 
quite common for 
the proportional and general knapsack problem to 
be called 
unweighted and weighted, respectively. The notion 
\emph{weight} is 
ambiguous, however, as some authors~\cite{BKKR2014} 
use it for what 
is called size here, while others~\cite{IZ2010} use 
it for what is 
called the value here or profit elsewhere. For the 
sake of clarity, 
we are well advised to avoid the term weight altogether.} In the 
generalized version, 
the goal is to maximize the total value of all 
packed items. 
With no further specification given, we are 
always referring to 
the proportional case.

A variant of the knapsack problem has been proposed by Iwama and
Taketomi~\cite{IT2002} under the name of \emph{removable} knapsack.
In this model, we can 
discard 
an item not 
only when
it is first presented to us; we may also remove a packed item from the
knapsack at any point. This is possible only once for each item, however;
once removed, an item cannot be repacked. As for the classical problem
without removability, the capacity of the knapsack may not be exceeded
at any point in time.   
Recently, Rossmanith has introduced a similar relaxed online setting for graph problems where decisions are taken only when constraints make it inevitable~\cite{Rossmanith2018}. 

This model is arguably just as natural a way to 
translate the 
knapsack problem into the online setting as the 
more well-examined 
variant without removability. In many cases, it 
will not be hard to 
discard items at regular intervals, only the chance 
of obtaining 
specific objects is subject to special 
circumstances. 
For a practical example, consider a storage room 
in which you can 
store all kinds of objects that you come across 
over time. In the 
beginning you can just keep collecting everything, 
but by doing so 
you inevitably run out of space before too long. 
Then you will have 
to start disposing of some of your possessions 
to make room for 
new, potentially more interesting acquisitions. 
Your goal is to end up with a selection of just 
the most meaningful 
and useful items that you could have. 
This paper analyzes the advice complexity of 
both the 
proportional and general removable knapsack 
problem. It is telling 
you how much information about upcoming 
opportunities you need to 
ensure an outcome that is either optimal or off 
by at most a given factor. 

\subsection{Advice Applications}
Besides inherently interesting insights into the 
information complexity of the knapsack problem, 
our advice algorithms also offer more 
concrete applications. 
Any algorithm reading a bounded number of advice bits can be implemented by running a bounded number of 
deterministic algorithms in parallel and selecting 
the best result. 
An advice analysis thus tells us, for example, 
how to optimally organize a betting pool in a
winner-take-all scenario. 
Our main result in particular provides 
a smart selection of strategies to be assigned 
to a mere constant number of actors 
such that one is guaranteed to be as close 
to optimality as we desire, 
no matter how difficult the instances of the 
general knapsack problem with removability may become. 

A further advantage of analyzing the advice complexity 
of a problem is that the resulting bounds are very versatile.  
The lower bounds are particularly strong. They show that 
a certain competitivity cannot be achieved with a given amount 
of additional information, regardless of the form this advice may take. 
The oracle is indeed able to convey to the algorithm all kinds of 
structural information about the adversarial instance; 
for example, in the case of our knapsack problem, 
whether items smaller than a given 
threshold should or must not be ignored, 
whether replacing packed items by later ones will ever be beneficial, 
whether the values span more than a certain range, 
whether an optimal solution fills the knapsack completely, 
whether there are multiple optimal solutions, and so on. 
Lower bounds on the competitivity of advice algorithms 
imply lower bounds for randomized algorithms, 
and our results indeed improve upon the best bounds known 
for randomization; \cref{thm:general_one_bit_lower} even completely 
closes the remaining gap in the analysis of barely random 
algorithms for the general knapsack problem.

There are also interesting implications for deterministic algorithms. 
Consider the multiple knapsack problem in which every item is either 
rejected or packed into one of $k>1$ knapsacks; 
the goal is for the algorithm to have in the end one knapsack that is as full as possible. 
This problem has been analyzed with removability by Iwama and Taketomi in 
the proportional case. 
In the conclusion of their paper~\cite{IT2002}, they pose it as an open 
problem to analyze this model if we are allowed to copy items and pack them 
into arbitrarily many of the available knapsacks. 
It turns out that deterministic algorithms for this problem 
with different $k$s are equivalent to advice algorithms: 
An advice algorithm restricted to $\log k$ advice bits can read up to $k$ 
different advice strings. Even if the algorithm reads the entire advice string 
right at the beginning, before taking any decision, it will thus implement 
one of $k$ deterministic strategies. 
Having $k$ knapsacks and being able to pack each item into several of them 
at the same time means that we can just simulate each possible strategy 
in one of the knapsacks and see which one leads to the best result in the end. 
Conversely, the oracle in our model already knows the optimal choice and 
can communicate to an advice algorithm which of the knapsacks it should be simulating. 
Knowing the advice complexity of our problem with only one knapsack therefore 
automatically yields a comprehensive competitive analysis for this deterministic 
problem for any $k>1$. 
All of this remains true for the general knapsack problem; thus our results 
provide a comprehensive picture for the proposed model in both the proportional 
and non-proportional case. 
We remark that algorithms for the resource buffer model are generally not 
applicable here. Algorithm~$2$ by Han et al.~\cite[Thm.~9]{HKMY2019}, for example, 
keeps regrouping items in every step and thus crucially relies on having 
a single large buffer instead of multiple standard-sized knapsacks without an 
option to shuffle items between them.

\section{Preliminaries}
Throughout this paper, $\log$ denotes the binary logarithm. 
We formally define the removable knapsack problem as follows. 
\begin{definition}[Removable Knapsack Problem]
\GenRemKnap is an online maximization problem. 
An instance $\inst$ is a sequence 
$(\size_1,\val_1),\dots,(\size_n,\val_n)$ of $n$ items, 
each of 
which is a pair of some real positive \emph{size} 
$\size_i$ and 
\emph{value} $\val_i$. 
Where useful, we denote size and value of an item 
$i$ functionally 
by $\size(i)= \size_i$ and $\val(i)=\val_i$. 
The domain of this function naturally extends to 
arbitrary subsets 
$\appsol\subseteq\{1,\dots,n\}$ of items by defining 
$\size(\appsol)=\sum_{i\in\appsol}\size(i)$ and 
$\val(\appsol)=\sum_{i\in\appsol}\val(i)$. 
The knapsack has a maximum size capacity, which we 
normalize to be 
$1$. 

The instance is presented item by item to an online algorithm $\appalg$ that has to maintain a \emph{packing}, a set of packed items. We call the total size of the currently packed items the current \emph{filling} of the knapsack. 
The algorithm starts out with an empty knapsack, represented by the empty set $\appsol_0=\emptyset$. 
When presented with item $i$, the algorithm may 
first remove any of 
the items $\appsol_{i-1}$ packed so far; then it 
may pack the new 
item if this does not exceed the knapsack 
capacity. In other words, 
the algorithm selects a subset 
$\appsol_i\subseteq\appsol_{i-1}\cup\{i\}$ with 
$\size(\appsol_i)\le 1$ in step $i$. 
The algorithm learns the size of item $i$ only 
once it is presented 
and only learns the total number $n$ of items 
after selecting 
$\appsol_n$. 
The final packing computed by $\appalg$ is denoted 
by 
$\appsol=\appsol_n$. 
The \emph{gain} that we aim to maximize is the 
total value 
$\val(\appsol)$ of the final packing. 

The proportional variant \PropRemKnap additionally 
satisfies 
$\size_i=\val_i$ for each item $i$.
\end{definition}

\begin{definition}[Competitive Ratio]
Let an online maximization problem with instance set $\instset$ be given and let $\appalg$ be an online algorithm solving it. 
For any instance $\inst\in\instset$, denote by $\app(\inst)$ the gain that $\appalg$ achieves on $\inst$ and by $\opt(\inst)$ the gain of an optimal solution to $\inst$ computed offline. 
The \emph{competitive performance} of $\appalg$ on an instance $\inst\in\instset$ is $\opt(\inst)/\app(\inst)$. 
For any $\rho\in\R$, the algorithm $\appalg$ is called \emph{strictly $\rho$-competitive} if it performs $\rho$-competitively across all instances, that is, if 
$\forall\,\inst\in\instset\colon\ \opt(\inst)/\app(\inst)\le \rho$.
The infimal competitivity $\textnormal{inf}\{\,\rho\in\R\mid \appalg\text{ is strictly $\rho$-competitive}\,\}$
is called \emph{strict competitive ratio} of $\appalg$.
We can weaken the defining inequality above so that it only needs to hold asymptotically in the sense of 
$\exists\,\alpha\in\R_+\colon\ \forall\,\inst\in\instset\colon\  
\opt(\inst)\le \rho\cdot\app(\inst)+\alpha$.
If this condition is met, we call $\appalg$ \emph{nonstrictly $\rho$-competitive}. 
\end{definition}

Note that strict $\rho$-competitivity implies nonstrict $\rho$-competitivity but not vice versa, making it harder to prove lower bounds for nonstrict competitivity. 
For the knapsack problem, however, it makes sense to always analyze competitivity in the strict sense: 
On the one hand, we obtain a nonstrict lower bound from a strict one by scaling up the knapsack capacity and all item sizes in a hard instance set such that the smallest item is strictly larger than $\alpha$. 
If, on the other hand, scaling is impossible due to the problem being defined with a fixed knapsack capacity of $1$, for example, then choosing $\alpha=1$ shows any online algorithm is $1$-competitive in the nonstrict sense.

\section{Related Work}
Knapsack is one of the 21 \textnormal{NP}-complete decision problem in Karp's famous list~\cite{Kar1972}. 
An algorithm based on dynamic programming solves both the proportional and the general version in pseudo-polynomial time; see Bellman~\cite[Section~1.4]{B57} for the general technique and Dantzig~\cite[p.~275]{D57} for a concrete description of its application to the knapsack problem.
The pseudo-polynomial time algorithm can be adapted to the optimization version, 
yielding a fully 
polynomial-time approximation 
scheme~\cite{IK1975}. 
In the following two subsections, we list the known 
results on the 
advice complexity of the proportional knapsack 
problem, first for 
the classical version and then for the variant 
allowing the removal 
of packed items. 

\subsection{Knapsack without Removability}\label{sec:relatedworkproportional}
Marchetti-Spaccamela and Vercellis were the 
first to consider the 
classical online version of the knapsack problem 
in 1995. They 
called it the \emph{$\{0,1\}$ knapsack problem} to 
distinguish it 
from the \emph{fractional} knapsack problem, which 
allows for 
packing items partially. 
They proved that both versions have an unbounded competitive ratio if items are allowed to have sizes different from their values~\cite[Thm.~2.1]{MV1995}. 
We denote the classical 
problem with 
neither fractional items nor removability by 
\Knap and its 
proportional variant by \PropKnap. 

The concept of advice emerged much later. When it did, \Knap quickly became one of the prime examples of a problem with an interesting advice complexity. 

First, just a single advice bit brings with it a jump from non-competitivity to a {$2$-competitive} algorithm~\cite[Thm.~4]{BKKR2014}.
More advice bits do not help however, as long as the number stays below $\lfloor\log(n-1)\rfloor$~\cite[Thm.~5]{BKKR2014}.
Once this threshold is surpassed, logarithmic advice allows for a competitive ratio that is arbitrarily close to 1~\cite[Thm.~6]{BKKR2014}. 
Achieving optimality, finally, requires at least $n-1$ advice bits~\cite[Thm.~3]{BKKR2014}.  

The situation for the general variant is simpler: 
Any algorithm 
reading less than $\log n$ advice bits has an 
unbounded competitive 
ratio, but $\mathcal{O}(\log n)$ advice suffices for 
a near-optimal 
solution~\cite[Thms.~11~and~12]{BKKR2014}. 
A schematic plot of the advice complexity 
behaviors just described 
can be found in 
Figure~\ref{fig:prop_schematics_combined} in light gray.

\subsection{Online Knapsack Variants}\label{sec:relatedworkgeneral}
Iwama and Taketomi~\cite{IT2002} proposed the online knapsack model with removability as it is examined in the present paper.  They proved that the competitive ratio for the proportional variant of this problem, which we denote by \PropRemKnap, is exactly 
the golden ratio. 
Iwama and Zhang later considered the problem with resource augmentation, that is, for online algorithms that may use a larger knapsack than the offline algorithm~\cite{IZ2010}. 

Later still, Han et al.~\cite{HKMG2014} proved an upper bound of $5/3$ on the competitive ratio for a variant of \PropRemKnap where the value $\val$ of an item is not necessarily proportional to its size $\size$ but not arbitrary either; instead, the value is given by a convex function $\val=f(\size)$ known to the algorithm. They also proved the golden ratio to be optimal if $f$ has some further technical properties.
Han et al.~\cite{HKM2014} considered online knapsack with removal costs, a variant of \PropRemKnap where items can be removed, but not for free. 

Noga and Sarbua~\cite{NS05} considered a knapsack 
variant, where it 
is possible to split each arriving item in two 
parts of not 
necessarily equal size, and combine this with 
resource 
augmentation. 
Han and Makino~\cite{HM2010} considered another partially fractional variant of \PropRemKnap where each item can be split a constant number of times at any time. 
Most importantly in our context, Han et al.~\cite{HKM2015} examined randomized algorithms for \PropRemKnap, proving an upper bound of $10/7$ and a lower bound of $5/4$ on the expected competitivity. 
Cygan et al.~\cite{CJS2016} extended the 
study of randomization for \PropRemKnap to a variant with 
multiple knapsacks. 
Recently, Böckenhauer et al.~\cite{BBHLR21} have 
introduced a new 
model for the online proportional knapsack 
problem in which items 
can be stored outside of the knapsack until the 
instance ends after 
paying a reservation fee that is a fixed 
fraction $\alpha$ of the 
item's value. 

\section{Results for Proportional Removable 
Knapsack}
In Section~\ref{sec:prop_opt_lower}, we consider how much---or rather, how little---removability helps when trying to obtain an optimal solution. 
In Section~\ref{sec:prop_one_bit}, we prove upper and lower bounds on what is possible with a single advice bit. 
Finally, we prove in Section~\ref{sec:prop_near_opt_upper} 
that a constant amount of 
advice is sufficient to achieve a competitive 
ratio of $1+\eps$, for an arbitrary 
$\eps>0$, and a constant depending on $\eps$. 
See Figure~\ref{fig:prop_schematics_combined} for a rough 
representation of these 
results in dark gray.

\begin{figure}
\begin{tikzpicture}[yscale=.9,xscale=1]
\begin{axis}[
x tick label style={rotate=0,anchor=center,yshift=-7.2},
y tick label style={rotate=0,anchor=east,xshift=1},
axis x line = bottom,
axis y line = left,
xmin=0,ymin=0,
xmax=25,ymax=11,
enlargelimits = false,
xtick align = outside,
ytick align = outside,
width=\textwidth-2em,
height=4.8cm,
xlabel={Advice Bits},
x label style={rotate=0,at={(axis description cs:1.005,-0.04)},anchor=south},
ylabel={Competitivity},
y label style={rotate=-90+0,at={(axis description cs:0.2,1.1)},anchor=near ticklabel},
ytick={0,2,6,8},
yticklabels = {$1$,$1+\eps$,$\Phi$,$2$},
xtick={0,2,12,22},
extra x ticks={9,21},
xticklabels = {$\phantom(0\hphantom)$,$O(1)$,$O(\log n)$,$\hphantom(n\hphantom)$},
extra x tick labels = {$\hphantom(\log 
n\hphantom)$,$\hphantom(n-\log n\hphantom)$}, 
extra x tick style = {
x tick label style={rotate=0,anchor=east,xshift=8.6,yshift=.35},
},
]
\addplot [black,fill=black,opacity=0.15] coordinates {
(0,10.2) (0.2,9.6) (0.6,10.4) (1,9.6) (1.4,10.4) (1.8,9.6) (2,10.2) (2,8) (9,8) (12,2) (21,2) (22,0)
}
|- (0,0) -- cycle;
\addplot [black,fill=black,opacity=0.3] coordinates {
(0,6) (2,2) (21,2) (22,0)
}
|- (0,0) -- cycle;
\end{axis}
\node[opacity=0.3] at (4.9,2.6) {Proportional 
Knapsack 
without 
Removability \dots};
\node[opacity=0.6] at (7.8,0.85) {\dots{}\;and with 
Removability};
\end{tikzpicture}
\begin{tikzpicture}[yscale=.9,xscale=1]
\begin{axis}[
x tick label style={rotate=0,anchor=center,yshift=-7.2},
y tick label style={rotate=0,anchor=east,xshift=1},
axis x line = bottom,
axis y line = left,
xmin=0,ymin=0,
xmax=25,ymax=11,
enlargelimits = false,
xtick align = outside,
ytick align = outside,
width=\textwidth-2em,
height=4.8cm,
xlabel={Advice Bits},
x label style={rotate=0,at={(axis description 
cs:1.005,-0.04)},anchor=south},
ylabel={Competitivity},
y label style={rotate=-90+0,at={(axis description 
cs:0.2,1.1)},anchor=near ticklabel},
ytick={0,2,8},
yticklabels = {$1$,$1+\eps$,$2$},
xtick={0,2,12,22},
extra x ticks={9,21},
xticklabels = {$\phantom(0\hphantom)$,$O(1)$,$O(\log 
n)$,$\hphantom(n\hphantom)$},
extra x tick labels = {$\hphantom(\log 
n\hphantom)$,$\hphantom(n-\log n\hphantom)$}, 
extra x tick style = {
x tick label 
style={rotate=0,anchor=east,xshift=8.6,yshift=.35},
},
]
\addplot [black,fill=black,opacity=0.15] coordinates {
(0,10.2) (0.2,9.6) (0.6,10.4) (1,9.6) (1.4,10.4) (1.8,9.6) 
(2.2,10.4) (2.6,9.6) 
(3.0,10.4) (3.4,9.6) (3.8,10.4) (4.2,9.6) (4.6,10.4) (5.0,9.6) 
(5.4,10.4) (5.8,9.6) 
(6.2,10.4) (6.6,9.6) (7.0,10.4) (7.4,9.6) (7.8,10.4) (8.2,9.6) 
(8.6,10.4) (9.0,9.6) 
(9.4,10.4) (9.8,9.6) (10.2,10.4) (10.6,9.6) (12,2) 
(21,2) (22,0)
}
|- (0,0) -- cycle;
\addplot [black,fill=black,opacity=0.3] coordinates {
(0,10.2) (0.2,9.6) (0.6,10.4) (2,2) (21,2) (22,0)
}
|- (0,0) -- cycle;
\end{axis}
\node[opacity=0.3] at (7.6,2.6) {General Knapsack 
without \dots};
\node[opacity=0.6] at (7.8,0.85) {\dots{}\;and with 
Removability};
\end{tikzpicture}
\caption{A schematic plot of the advice complexity behavior of the classical online knapsack problem in light gray and the relaxed variant with removability in dark gray. 
For the proportional version without 
removability there are two 
large plateaus; removability collapses to a single 
vast expanse.
For the general version, in which an item's value 
may differ from 
its size, there is only one but a more extreme 
jump directly from 
an unbounded competitive ratio to near 
optimality; with 
removability, this jump is occurring earlier and 
even steeper. 
}
\label{fig:prop_schematics_combined}
\end{figure}

\subsection{Achieving Optimality}\label{sec:prop_opt_lower}
We begin by briefly considering \PropKnap, the classical proportional knapsack problem without removability. 
Solving it optimally is trivial with $n$ advice bits: The algorithm reads one bit per item, telling it whether to  accept or reject.  
\Cref{{thm:classic_prop_opt_lower_tight}} proves this to be tight by lifting the best known lower bound from $n-1$ advice bits~\cite[Thm.~3]{BKKR2014} to $n$ advice bits. 
\begin{theorem}\label{thm:classic_prop_opt_lower_tight}
Any algorithm for \PropKnap reading less than $n$ 
advice bits is 
suboptimal.
\end{theorem}
\begin{proof}
For every $n$, we consider the $2^n$ instances that all begin with the same $n-1$ items, namely one for each of the sizes $s_1=2^{-1},s_2=2^{-2},\dots,s_{n-1}=2^{-(n-1)}$.
The size of the final item is either $s_n=2^{-n}$ or any $s_n'\in\{\,1-\sum_{i\in I}s_i\mid I\subsetneq\{1,2,\dots,n-1\}\,\}$. 
We remark that this hard instance family is almost identical to the one used by Böckenhauer et al.~\cite[Thm.~3]{BKKR2014} for proving the lower bound of $n-1$ advice bits; 
the only tiny modification is changing the item size $s_n'=1-\sum_{i\in I}s_i$ for $I=\{1,2,\dots,n-1\}\}$ to $s_n=2^{-n}$. 
There are $2^{n-1}$ options for the final item, and each one requires an optimal algorithm to have packed another one of the $2^{n-1}$ possible subsets of the previous $n-1$ items: 
If the last item turns out to have size $s_n'=\sum_{i\in I}s_i$, for any given $I\subsetneq \{1,2,\dots,n-1\}$, then the knapsack can be filled completely if and only if exactly the items $I$ have been packed before it. 
And if the last item has size $s_n=2^{-n}$, then the algorithm needs to pack all items for optimality, leaving a gap of $2^{-n}$. 
Now, an advice string of length $n-1$ should already enable us to distinguish between the $2^{n-1}$ possible sizes of the last item, allowing the algorithm to select the optimal subset of the first $n-1$ items. 
An advice algorithm reading fewer than $n-1$ advice bits on instances of length $n$ inevitably packs the same subset of the first $n-1$ items for two different sizes of the last item, leading to a suboptimal performance for one of the two. 

This is essentially what led to the previously known lower bound of $n-1$ advice bits for optimality. 
The key point to note now is that the algorithm has to read these $n-1$ advice bits, proved necessary for optimality on instances of length $n$, before the last item is presented.
Thus it is in fact reading $n-1$ advice bits while processing the first $n-1$ items of the instance, which---thanks to our modification---also constitute a complete instance on their own. 
Thus $n'=n-1$ advice bits are consumed on an instance of length $n'$, concluding the proof. 
\end{proof}
Having determined \PropKnap's 
advice complexity for optimality, we now do the same for \PropRemKnap, the variant with removability. It turns out that the option to remove items hardly helps at all in achieving optimality.
We begin the upper bound, which is simple but instructive as to what is possible with removability. 
\begin{theorem}\label{thm:prop_opt_upper}
There is an optimal algorithm for \PropRemKnap reading $n-1$ advice bits.
\end{theorem}

\begin{proof}
Consider an algorithm that packs the first item without reading any advice bits. 
For each subsequent item, it reads one advice bit, telling it whether the new item is part of a fixed optimal solution. 
If so, then the new item is packed; otherwise, it is rejected. 
The first item, which has been packed without advice, is kept in the knapsack as long as there is enough room for it. 
If the first item is part of the fixed optimal solution, then it will always fit in beside the other items being packed; 
otherwise, it will be discarded at some point. 
Thus the algorithm is able to reproduce the fixed optimal solution exactly. 
\end{proof}

\begin{theorem}\label{thm:prop_opt_lower}
Solving \PropRemKnap optimally requires more than $n-\log n$ advice bits.
\end{theorem}

\begin{proof}
Let any positive $\eps<1/3$ of the form $\eps=1/(2j)$ for an odd integer $j$ be given. Now, we can choose an arbitrarily large odd integer $m$ such that $2\eps m+1$ is a power of two, namely an appropriate multiple of $j$. 
Since $2\eps m+1$ is even and $m$ is odd, we know that $m-2\eps m$ is even and thus $m/2-\eps m$ is an integer. We denote it by $k=m/2-\eps m$ and note that $m-2k=2\eps m$.
We consider a 
family of instances that are all identical, with 
exception of the final item. 
The capacity of the knapsack shall be $m-k+1$. We 
could of course 
normalize this to $1$ by scaling down the capacity 
and all item sizes. 

First, for each $i\in\{0,1,\dots,\log(2\eps m+1)-1\}$, 
an item of size $2^i$ is presented. 
Note that these $\log(\eps m +1)$ items have a total size of $\sum_{i=0}^{\log(2\eps m+1)-1}2^i=2^{\log(2\eps m+1)}-1=2\eps m$. 
Moreover, for every $j\in\{0,1,\dots,2\eps m\}$, there is exactly one item subset of total size $j$. 

The instance then continues with $m$ more items, one 
of each size in 
\[M=\Bigl\{1+\frac12,\ldots,1+\frac1{2^m}\Bigr\}=\Bigl\{\,1+\frac1{2^i}\Bigm|
 1\le i\le m\,\Bigr\}.\] 

Finally, a single item of size $1-\sum_{s\in S}(s-1)$ is 
presented, where $S$ 
may be any subset of $M$ whose cardinality satisfies $k\le |S|\le m-k$. 
There are $\sum_{i=k}^{m-k}\binom mi$ such subsets 
and we have one 
instance for each possible choice. 
A straightforward tail bound, which we derive 
at the end of this proof, 
shows that we have 
$\sum_{i=0}^{k}\binom mi\le 2^{mH(1/2-\eps)+\log m-1}$, 
where 
$H(p)=-p\log p-(1-p)\log (1-p)$ is the binary 
entropy function.
We also know that $\sum_{i=0}^m\binom mi=2^m$ and 
therefore obtain 
\begin{align*}
\sum_{i=k}^{m-k}\binom mi&\ge 2^m-2\cdot 
2^{mH(1/2-\eps)+\log 
m-1}
=2^m(1-2^{-m(1-H(1/2-\eps))+\log m})
\end{align*} as a lower 
bound on the 
number of instances. 

Consider the instance whose final item has size 
$1-\sum_{s\in S_0}(s-1)$ for 
any given $S_0\subseteq M$. 
The knapsack of capacity $m-k+1$ can be filled completely 
with the items of 
this instance as follows: 
Pack the last item, all items of sizes in $S_0$---which brings us to a total size of exactly $|S_0|+1$---and finally the unique subset of the first $\log(2\eps m+1)$ items with a total size of $m-k+1-(|S_0|+1)=m-(k+|S_0|)$. 
This is in fact the only way to fill the knapsack entirely for the following reason. There is exactly one selection of the items with non-integer sizes such that their fractional parts sum up to an integer---namely exactly $1$---which is necessary to fill the knapsack of integer capacity completely. 
The remaining integer gap has to be filled by the remaining items with integer sizes, and we have already noted before that there is exactly one way to do this for every integer between $0$ and $m-2k=2\eps m$. 
The instances are all identical until the 
last item is 
presented. 
Therefore, any online algorithm must have packed exactly the right selection of items when the final one is presented to realize the optimal solution. 
The number of advice bits necessary to 
distinguish the possible instances and guarantee optimality is therefore  
\begin{align*}
\log\sum_{i=k}^{m-k}\binom mi&\ge 
m+\log(1-2^{-m(1-H(1/2-\eps))+\log m}).
\end{align*} 
The following straightforward calculations---using the Taylor expansion for the logarithm and Stirling bounds---show that this can be bounded from below by $n-\log(5\eps n)>n-\log(n)+\log(1/\eps)-3$ for sufficiently large $m$.

Using the Taylor expansion for the natural logarithm around $1$, we have 
\begin{align*}
\log(1-x)
&=\frac{-1}{\log e}\sum_{i=1}^\infty \frac{x^i}i\\
&>-x\sum_{i=1}^\infty \frac{x^{i-1}}{i}\\
&=-x\sum_{i=0}^\infty \frac{x^i}{i+1}\\
&>-x\sum_{i=0}^\infty x^i\\
&>-x\frac 1{1-x}\\
&\ge-2x
\end{align*} 
for $0<x<1/2$. 
Applying this to $x=2^{-m(1-H(1/2-\eps))+\log m}$ we can thus further bound the number of required advice from below by 
\begin{align*}
\log\sum_{i=k}^{m-k}\binom mi&\ge
m-2\cdot2^{-m(1-H(1/2-\eps))+\log m}.
\end{align*} 
Each of the described instances contains $n=
\log(2\eps m+1)+m+1$ items. 
This is greater than the number of required advice bits by at most 
\[\log(2\eps m +1)+1+2\cdot 2^{-m(1-H(1/2-\eps))+\log m}.\]
For any fixed, positive $\eps<1/2$, we have $0<H(1/2-\eps)<1$ and thus the exponent $-m(1-H(1/2-\eps))+\log m$ growing to arbitrarily large negative numbers for increasing $m$. This means that the difference between the number of items $n$ and the number of required advice bits asymptotically coincides with $\log(2\eps m+1)+1=n-m$.
 
We have $m=n-\log(2\eps m +1)-1> n-\log(2\eps n+1)-1$ and thus $n-m< \log(2\eps n+1)+1<\log (5\eps n)$ 
for sufficiently large $m$ and thus $n$. 
Hence more than $n-\log(5\eps n)>n-\log(n)+\log(1/\eps)-3$ advice bits are necessary for optimality on instances of length $n$ for an arbitrarily small positive $\eps<1/2$.

It remains to prove the mentioned tail bound, for which we use the standard Stirling 
bounds~\cite{R55}, which---in a simple form---are $\sqrt{2\pi n}(n/e)^n\le n!\le e\sqrt{n}(n/e)^n$.

We obtain 
\begin{align*}
\sum_{i=0}^k\binom mi
&= 1+\sum_{i=1}^k\binom mi\\
&\le 1+k\binom mk\\
&\le1+\frac{m!k}{k!(m-k)!}\\
&\le 
1+\frac{ek\sqrt{m}(m/e)^m}{\sqrt{2\pi k}(k/e)^{k}\sqrt{2\pi(m-k)}((m-k)/e)^{m-k}}\\
&= 
1+\frac {ek}{2\pi}\frac{\sqrt{m}}{\sqrt{k(m-k)}}\frac{m^m}{k^{k}(m-k)^{m-k}}\frac{e^ke^{m-k}}{e^m}\\
&= 
1+\frac {ek}{2\pi}\sqrt{\frac{m}{k(m-k)}}\left(\frac{m}{k}\right)^{k}\left(\frac{m}{m-k}\right)^{m-k}\\
\intertext{Since $k=m/2-\eps m=m(1/2-\eps)$, we have $m-k=m(1/2+\eps)$. Using this and the binomial theorem $(1/2-\eps)(1/2+\eps)=1/4-\eps^2$ twice, we obtain the following bound.}
\sum_{i=0}^k\binom mi&\le 
1+\frac {ek}{2\pi}\frac1{\sqrt{m(1/4-\eps^2)}}\left(\frac1{1/2-\eps}\right)^{m(1/2-\eps)}\left(\frac1{1/2+\eps}\right)^{m(1/2+\eps)}\\
\intertext{Since $H(p)=p\log(1/p)+(1-p)\log(1/ (1-p))$, we have  
$2^{H(p)}=\left(\frac1p\right)^p\left(\frac1{1-p}\right)^{1-p}$ and thus }
\sum_{i=0}^k\binom mi&\le 
1+\frac {ek}{2\pi}\frac1{\sqrt{m(1/4-\eps^2)}}2^{mH(1/2-\eps)}.\\
\intertext{Since $e/(2\pi)>1/3$ and $\eps\le 1/3$, we can continue as follows.}
\sum_{i=0}^k\binom mi&
\le1+3k\frac{1}{\sqrt m/9}2^{mH(1/2-\eps)}\\
&\le1+\frac{9k}{\sqrt m}2^{mH(1/2-\eps)}\\
&=1+9\sqrt{m}(1/2-\eps)2^{mH(1/2-\eps)}\\
&\le1+(9/2)\sqrt{m}\cdot 2^{mH(1/2-\eps)}-9\eps\sqrt{m}\cdot 2^{mH(1/2-\eps)}\\
\intertext{For any given, positive $\eps<1/2$, we have $H(1/2-\eps)>0$. 
Choosing $m\ge9^2$ sufficiently large yields the desired tail bound:}
\sum_{i=0}^k\binom mi&\le 1+(m/2)\cdot 2^{mH(1/2-\eps)}-9\eps\sqrt{m}\cdot 2^{mH(1/2-\eps)}\\
&\le (m/2)\cdot 2^{mH(1/2-\eps)}\\
&=2^{mH(1/2-\eps)+\log m-1}
\end{align*}
This concludes the proof of the theorem.
\end{proof}

\subsection{A Single Advice Bit}\label{sec:prop_one_bit}
The previous section covered the upper end of the advice spectrum, showing that, asymptotically, reading one advice bit for each item in the instance is necessary and sufficient for ensuring an optimal solution. 
We now turn to the other extreme and ask what can be done with the least nonzero amount of advice, one single bit for the entire instance. 

First, we describe a very simple $3/2$-competitive advice algorithm where a single advice bit indicates whether there is an optimal solution containing more than one item from the interval $[1/3,2/3]$: 
If the answer is yes, the algorithm maintains the smallest item in this interval until a second item fits in, while ignoring all items outside of the interval. As soon as a second item fits, it is packed and all remaining items are rejected. 
If the answer is no, the algorithm maintains in the knapsack the largest item of size at least $1/3$ seen so far while packing all items smaller than $1/3$ as long as they fit. 
If the knapsack capacity is never exceeded, the solution is optimal. 
If the knapsack capacity is exceeded at some point, all packed items but possibly one are smaller than $1/3$. Discard these items one by one, in arbitrary order, until we are within the capacity of the knapsack again. The remaining gap is at most $1/3$.

Han et al.~\cite[Thm.~6]{HKM2015} have presented a randomized algorithm that relies on a partition of the items into six size classes.  It is rather involved and hard to analyze, yet yields an expected competitive ratio of $10/7\approx1.428571$. Because it uses only a single random bit, it provides an upper bound for our case of one advice bit as well. 
We undercut this bound with a more manageable $\sqrt{2}$-competitive algorithm that needs only five classes.  We then complement this with a lower bound of $(1+\sqrt{17})/4=4/(\sqrt{17}-1)\approx 1.2808$.
\begin{theorem}\label{thm:prop_one_bit_upper}
There is a $\sqrt{2}$-competitive algorithm for \PropRemKnap reading only one advice bit. 
\end{theorem}
\tikzset{every picture/.style={thick,rounded corners=0pt,line cap=round}}
\begin{figure}[H]
\centering
\begin{tikzpicture}
\newcommand{\myheight}{.22}
\newcommand{\mywidth}{*13.8}
\newcommand{\myadjacent}{.3}
\newcommand{\mygap}{.0072\mywidth}
\newcommand{\posA}{0.29289\mywidth}
\newcommand{\posB}{0.41421\mywidth}
\newcommand{\posC}{0.5\mywidth}
\newcommand{\posD}{0.70711\mywidth}
\newcommand{\posE}{1\mywidth}
\newcommand{\closedopen}[2]{
\draw (#2-\mygap,-\myheight) -- (#1,-\myheight) -- (#1,\myheight) -- (#2-\mygap,\myheight);
\draw (#2-\mygap,-\myheight) arc (-atan(\myheight/\myadjacent):atan(\myheight/\myadjacent):{sqrt(\myheight^2+\myadjacent^2)});
}
\newcommand{\openclosed}[2]{
\draw (#1+\mygap,-\myheight) -- (#2,-\myheight) -- (#2,\myheight) -- (#1+\mygap,\myheight);
\draw (#1+\mygap,\myheight) arc (180-atan(\myheight/\myadjacent):180+atan(\myheight/\myadjacent):{sqrt(\myheight^2+\myadjacent^2)});
}
\openclosed{0}{\posA}
\openclosed{\posA}{\posB}
\openclosed{\posB}{\posC}
\openclosed{\posC}{\posD}
\openclosed{\posD}{\posE}

\node at (\posA/2,0) {\footnotesize tiny\vphantom{tinysmallmediumbighuge}};
\node at (\posA/2+\posB/2,0) {\footnotesize small\vphantom{tinysmallmediumbighuge}};
\node at (\posB/2+\posC/2,0) {\footnotesize medium\vphantom{tinysmallmediumbighuge}};
\node at (\posC/2+\posD/2,0) {\footnotesize big\vphantom{tinysmallmediumbighuge}};
\node at (\posD/2+\posE/2,0) {\footnotesize huge\vphantom{tinysmallmediumbighuge}};

\draw[decoration={brace}, decorate] (\posA+\mygap/2,.3) node {} -- (\posC,.3);
\draw[decoration={brace}, decorate] (\posC+\mygap/2,.3) node {} -- (\posE,.3);

\node at (\posA/2+\mygap/4+\posC/2,.6) {\footnotesize little\vphantom{littlelarge}};
\node at (\posC/2+\mygap/4+\posE/2,.6) {\footnotesize large\vphantom{littlelarge}};

\node at (0+\mygap,-.5) {$0$};
\node at (\posA+\mygap/4,-.5) {\vphantom{$0abcd1$}$a$};
\node at (\posB+\mygap/4,-.5) {\vphantom{$0abcd1$}$b$};
\node at (\posC+\mygap/4,-.5) {\vphantom{$0abcd1$}$c$};
\node at (\posD+\mygap/4,-.5) {\vphantom{$0abcd1$}$d$};
\node at (\posE-\mygap/4,-.5) {\vphantom{$0abcd1$}$1$};
\end{tikzpicture}
\caption{The partition of the interval $(0,1]$ of possible sizes into the five subintervals used in the proof of Theorem~\ref{thm:prop_one_bit_upper}---namely $(0,a]$, $(a,b]$, $(b,c]$, $(c,d]$, and $(d,1]$---plus the corresponding class names. 
The values are $a=1-1/\sqrt2\approx0.293$, and $b=\sqrt2-1\approx0.414$, and $c=1/2$, and $d=1/\sqrt2\approx0.707$.}
\label{fig:classes_prop}
\end{figure}

\begin{proof}
We split the interval $(0,1]$ of possible sizes into subintervals at four points $a<b<c<d$. We will call the items with sizes in one of these five intervals \emph{tiny}, \emph{small}, \emph{medium}, \emph{big}, and \emph{huge}, respectively. 
Formally, we partition the items into the five 
classes 
\begin{align*}
P_\winzig{}&=\{\,i\mid 0<\size(i)\le a\,\},& 
P_\klein{}&=\{\,i\mid a<\size(i)\le b\,\},&
P_\mittel{}&=\{\,i\mid b<\size(i)\le c\,\},\\
P_\gross{}&=\{\,i\mid c<\size(i)\le d\,\},\text{}&
P_\riesig{}&=\{\,i\mid d<\size(i)\le 1\,\},
\end{align*} 
where
$a=1-1/\sqrt2\approx 0.29289$, $b=\sqrt2-1\approx0.41421$, $c=1/2$, and $d=1/\sqrt2\approx0.70711$.

We will call the small and medium items the \emph{little} ones collectively and refer to the big and huge items as the \emph{large} ones. 
Accordingly, we let $P_\kleinlich=P_\klein\cup P_\mittel$ and $P_\groesslich=P_\gross\cup P_\riesig$. See Figure~\ref{fig:classes_prop} for an illustration of the subintervals and class names.

The oracle uses the one available advice bit to tell the algorithm which of the two strategies described below to apply. 
For the decision, the oracle picks an arbitrary optimal solution $S$ to the given instance. 
If $S$ contains a large item, the first strategy will be chosen, with one exception: If the instance contains no huge item but a little and a big item that fit into the knapsack together, then the first strategy is chosen only if a minimal big item appears in the instance before a minimal small item. 
In all other cases, the second strategy is implemented. 

\begin{description}
\item[Strategy One]
If at any point a huge item appears, the algorithm packs it and keeps it until the end, discarding everything else. 

Otherwise, the algorithm operates with two slots, a 
primary one and a secondary one. 
In the primary slot, it maintains the minimal big item and in the secondary slot it maintains the minimal little item. 
The primary slot takes precedence; that is, in case of a conflict where a new minimal item for one slot is presented that does not fit with the minimal item in the other slot, we discard the little item. 

While maintaining the slot contents, tiny items are always packed greedily. If at any point a presented tiny item does not fit, the current contents of the knapsack are frozen and kept as they are until the instance has ended. The same happens after a step in which only tiny items have been discarded.

\item[Strategy Two]
This strategy manages not only two but three slots, all of which maintain minimal items of some class. In order of precedence, the primary slot maintains two medium items, the secondary slot up to three small items, and the tertiary one big one. 
As an exception, if at any point a big item appears that can be packed alongside a currently packed small item by discarding everything else, then this is done and these two items are kept till the end. 
The tiny items are handled as before: They are packed greedily and if either a presented tiny item does not fit or only tiny items have been discarded in one step, then the current knapsack configuration is kept up to the very end. 
\end{description}

We now need to carefully work through a case 
distinction according to 
the conditions listed in 
Table~\ref{tab:caseconditions} and show 
that the algorithm's competitivity is indeed 
bounded from above by 
$\max\{1/d,d/c,$ 
$1/2b,1/(a+b),1/(1-a),b/a\}=\sqrt2$. 

\begin{table}[ht]
\caption{The mutually exclusive cases considered in Theorem~\ref{thm:prop_one_bit_upper}.}
\label{tab:caseconditions}
\begin{tabular}{ccccccc}
\toprule
Case &Strategy&Competitivity&\multicolumn{4}{c}{Case Conditions}\\
\cmidrule(rl){1-1}\cmidrule(rl){2-2}\cmidrule(rl){3-3}\cmidrule(lr){4-7}
A    &One&$1/d$&$|P_\riesig|>0$&      &         &     \\
B    &One/Two&$d/c$&$|P_\riesig|=0$& $|S\cap P_\gross|>0$& $|P_\mittel|\le1$ &     \\
C    &Two&$1/2b$&$|P_\riesig|=0$& $|S\cap P_\gross|\ge0$& $|P_\mittel|>1$ &     \\
D    &Two&$1/(a+b)$&$|P_\riesig|=0$& $|S\cap P_\gross|=0$& $|P_\mittel|=1$ &$|P_\klein|>0$ \\
E    &Two&$b/a$&$|P_\riesig|=0$& $|S\cap P_\gross|=0$& $|P_\mittel|=0$ &$|P_\klein|>0$ \\
F    &Two&$1/(1-a)$&$|P_\riesig|=0$& $|S\cap P_\gross|=0$& $|P_\mittel|\le1$ &$|P_\klein|=0$ \\
\bottomrule
\end{tabular}
\end{table}

\subparagraph*{\textit{Case A.}}
This case is trivial: If there are huge items in 
the instance, the 
first one will be packed and kept in the knapsack, 
yielding a 
competitive performance of $1/d$ or better.

\subparagraph*{\textit{Case B.}}
The case condition $|S\cap P_\gross|>0$ tells us 
that the optimal 
solution contains at least one big item. 
Since big items have a size above $c=1/2$, it 
contains exactly one. 
Moreover, there is at most one medium item in the 
entire instance. 
We now consider two subcases.

\emph{Subcase B1:} 
Assume first that the instance contains no pair 
of a little and a 
big item that can be packed at the same time. 
This means that the first strategy---which gives 
preference to big 
items over little ones---is operative. The 
algorithm will thus only 
discard a possibly packed little item when the 
first appearing big 
item appears to replace it. 
The strategy guarantees that one big item of 
size greater than $c$ 
is contained in the knapsack in the end. 
If there are no tiny items, both the online 
solution and the 
optimal solution $S$ contain exactly one big 
item and nothing else, 
implying a competitive performance of $c/d$. For 
the case that 
there are tiny items, recall that they are always 
packed greedily. 
Moreover, if any tiny item does not fit or is 
dismissed in some 
step, the algorithm conserves the current state 
of the knapsack, 
guaranteeing a filling of at least $1-a$. 
We may therefore assume that every tiny item is 
packed and kept in 
the online solution computed here. Whatever tiny 
items are 
contained in the optimal solution $S$ are thus 
in the online 
solution as well. If they have a total size $t$, 
the competitive 
performance is thus bounded from above by 
$(d+t)/(c+t)\le d/c$. 

\emph{Subcase B2:} Assume that there are a little 
and a big item 
that can be packed alongside each other. Let $i$ 
be the first 
minimal little item and $j$ the first minimal 
big item. Clearly, 
$i$ and $j$ fit into the knapsack together. If the 
knapsack 
contains these two items in the end, the 
competitive performance is 
$1/(a+c)\le 1/(a+b)=d/c$ or better. 

\emph{Subcase B2a:} Assume that $i$ appears after 
$j$. In this 
case, the first strategy is used. Since it 
maintains the smallest 
big item seen so far in the primary slot, it will 
have $j$ packed 
when $i$ is presented, allowing for $i$ to be 
packed alongside it. 

\emph{Subcase B2b:} Assume that $i$ appears before 
$j$, letting the 
algorithm implement the second strategy. 
The chosen strategy maintains up to two minimal 
medium items in the 
primary slot and up to three minimal small items 
in its secondary 
slot. 
However, there is by the global assumption of case 
$B$ at most one 
medium item in the entire instance and 
a small item will always fit in beside a medium 
item because 
$b+c\le 1$.  Hence, overall, a minimal little item is 
maintained in 
the knapsack, meaning that $i$ is packed already 
when $j$ appears, 
allowing for a little and big item to be packed 
together according 
to the strategy's stated exception.

\subparagraph*{\textit{Case C.}}
This simple case is covered by the second 
strategy. Maintaining two 
medium items in the slot with highest precedence 
guarantees a 
filled fraction of at least $2b$. 

\subparagraph*{\textit{Case D.}}
This case is easy as well: The one medium item 
will be packed into 
the primary slot when presented, and at least one 
small item is 
packed into and stays in the secondary slot 
because any small item 
fits in beside any medium item. 
The total size of small and a medium item is at 
least $a+b$, 
leading to a competitive performance of $1/(a+b)$ 
or better.

\subparagraph*{\textit{Case E.}}
Since there are no medium items in this case, the 
secondary slot 
will maintain up to three minimal small items. If 
there are three 
small items that fit together, they will 
eventually be packed, 
yielding a filled fraction of at least $3a>a/b$. 
Otherwise, two small items will be packed, or only 
one if and only 
if it is the only one. If there were no tiny items, 
we could in 
both cases bound the competitivity by $b/a$, 
using the maximal and 
minimal possible size of a small item. 
However, repeating the argument of subcase B1, we 
may assume that 
all tiny items are packed and none discarded, 
otherwise the packing 
would freeze instantly with a filling of at 
least $1-a$. 
If the tiny items have a total size of $t$, our 
bound would 
therefore only improve to $(b+t)/(a+t)\le b/a$. 

\subparagraph*{\textit{Case F.}}
This case is quite simple again. If there is a 
medium item, it is 
always packed and kept to the end. Beside this 
one potential medium 
item, there are only tiny ones, the minimization in 
the slot will 
therefore not affect the result adversely. If all 
items of the 
instance fit into the knapsack together, they are 
all packed and 
the solution is optimal. Otherwise, the greedy 
packing of tiny 
items leaves a gap of less than $a$, ensuring a 
competitive factor 
of $1/(1-a)$ or better. 
\end{proof}

We now complement the upper bound of \cref{thm:prop_one_bit_upper} with a lower bound of $(1+\sqrt{17})/4=4/(\sqrt{17}-1)\approx 1.2808$. 
\begin{theorem}\label{thm:prop_one_bit_lower}
No algorithm for \PropRemKnap reading only a single 
advice bit can 
have a better competitive ratio than 
$(1+\sqrt{17})/4$.
\end{theorem}

\begin{proof}
Let $\psi = 4/(1+\sqrt{17})$ 
and choose a positive $\eps<\psi\approx0.7808$. 
Let an algorithm for \PropRemKnap reading only a 
single advice bit 
be given. 
Consider the three instances $I_1$, $I_2$, and 
$I_3$ that all start 
with the same three items of sizes $x_1=\psi$, 
$x_2=\psi^2$, and 
$x_3=1-\psi^2+\eps$, which is the end of instance 
$I_1$ but 
followed by a last item of size $y_2=1-\psi^2$ 
for $I_2$ and of 
size $y_3=\psi^2$ for $I_3$. 
\begin{table}
\caption{A hard instance family for \PropRemKnap 
reading one advice 
bit; see the proof of Theorem~\ref{thm:prop_one_bit_lower}.}
\label{tab:prop_one_bit_upper}
\begin{tabular}{c@{\hspace{1.4em}}ccccccc@{\hspace{1.3em}}c@{\hspace{1.3em}}c}
\toprule
	&$x_1$&$x_2$&$x_3$&$y_2$&$y_3$&optimal&second 
	best&ratio\\
	\cmidrule(r){2-6}\cmidrule(lr{.9em}){7-9}
	$I_1$:&$\psi$&$\psi^2$&$1-\psi^2+\eps$&$$&&$\psi$&$\psi^2$&$\psi/\psi^2$\\
	$I_2$:&$\psi$&$\psi^2$&$1-\psi^2+\eps$&$1-\psi^2$&&$1$&$2(1-\psi^2)+\eps$&$1/(2(1-\psi^2)+\eps)$\\
	$I_3$:&$\psi$&$\psi^2$&$1-\psi^2+\eps$&&$\psi^2-\eps$&$1$&$\psi$&$1/\psi$\\
\bottomrule
\end{tabular}
\end{table}
For each $i\in\{1,2,3\}$, the instance $I_i$ has a 
unique optimal 
solution; it contains $x_i$ and, except for $i=1$, 
additionally 
$y_i$. 
Table~\ref{tab:prop_one_bit_upper} shows the total size 
for each of 
these optimal solutions and the second best 
solution. 

Because any two of these three items $x_1$, 
$x_2$, and $x_3$ sum up 
to over 1, the advice algorithm can keep at most 
one of them in the 
knapsack after the presentation of $x_3$. 
Moreover, since only one advice bit is given and 
the three 
instances are indistinguishable until after the 
decision on the 
third item has been taken, there are two instances 
for which the 
same item, if any, is kept in the knapsack for the 
presentation of 
the potential fourth item. 
This implies that the algorithm is suboptimal 
for at least one 
instance. The second best solutions for $I_1$, 
$I_2$, and $I_3$ 
fill up a fraction $\psi$, $2(1-\psi^2)+\eps$, and 
$\psi$, 
respectively. Thus, the competitive ratio of cannot 
be better than 
the minimum of $\psi/\psi^2$, $1/(2(1-\psi^2)+\eps)$, 
and $1/\psi$. 
Since $2(1-\psi^2)=\psi$, this means the 
competitivity is 
$1/(\psi+\eps)$ at best for arbitrarily small 
$\eps$. 
\end{proof}

Note again that an advice bit is at least as powerful as a random bit, 
hence \cref{thm:prop_one_bit_lower} also improves the 
best known lower bound 
of $5/4$ for one random bit due to Han et 
al.~\cite[Thm.~8]{HKM2015}. 

\subsection{Near Optimality with Constant Advice}\label{sec:prop_near_opt_upper}
Having seen how much advice is necessary for optimality and what the effect of a single advice bit can be, we now address the entire range in between. 
For this, we prove the following generalization of \cref{thm:prop_one_bit_lower}. 

\begin{theorem}\label{thm:prop_const_bit_lower}
Let an arbitrary integer $k>1$ be given. No 
algorithm for \PropRemKnap reading at most $\log 
k$ advice bits can achieve a better competitive 
ratio than $4/(3-2k+\sqrt{4k(k+1)-7})$.
\end{theorem}

\begin{proof}
We generalize the hard instance family from the 
proof of 
Theorem~\ref{thm:prop_one_bit_lower}. 
Let an arbitrary integer $k>1$ be given and 
define $\zeta$ as the 
positive root of $2\zeta^2+(2k-3)\zeta-2(k+1)$, 
namely 
$\zeta=(3-2k+\sqrt{4k(k+1)-7})/4$.
Consider $k+1$ instances that all start with the 
same $k+1$ items 
of the following, decreasing sizes: first $x_1=\zeta$, 
then 
$x_i=\zeta^2-(i-2)(1-\zeta)$ for every 
$i\in\{2,\dots,k\}$, and 
then $x_{k+1}=\zeta^2+(k-1)(1-\zeta)+\eps$ for an 
arbitrary $\eps$ 
satisfying $0<\eps<1-\zeta$. The instance $I_1$ 
ends immediately 
after these common items, whereas the instances 
$I_i$, for 
$i\in\{2,\dots,k+1\}$, presents one additional item 
of size 
$y_i=1-x_i$ as the final one.
There is a unique optimal solution for each 
instance: For $I_1$, it 
is to pack the first item of size $x_1=\zeta$. For 
$I_i$ with 
$i>1$, it is to pack the item of size $x_i$ and 
the last one of 
size $y_i=1-x_i$, which sum up to the optimal 
solution value $1$. 
Since there are only $\log k$ advice bits 
available to handle the 
$k+1$ instances, at least two instances $I_i$ and 
$I_j$ with $i<j$ 
are processed with the same advice string and 
thus the same 
deterministic algorithm. Consider this algorithm 
and the moment 
after seeing and taking decisions on the first 
$k+1$ items. 
It is impossible for the algorithm to have more 
than one of these 
common items packed since the two smallest of 
them already have a 
combined size of 
$x_k+x_{k+1}=2\zeta^2+(2k-3)\zeta-2k+3+\eps=1+\eps$. 
Now, if item $i$ is packed at the considered 
moment, the algorithm 
will perform suboptimally on instance $I_j$. 
Analogously, if item 
$j$ is packed, the performance on instance $I_i$ 
is suboptimal.

Now if suffices to check that the best 
suboptimal solution has a 
filling of at most $\zeta^2$ for $I_1$, at most 
$\zeta$ for $I_2$, 
\dots, $I_k$, and at most $\zeta+\eps$ for $I_{k+1}$. 
This leads to 
a performance ratio that is $\zeta/\zeta^2=1/\zeta$ or 
$1/(\zeta+1)$ at best, depending on the concrete 
algorithm, thus 
proving the theorem. 
See Table~\ref{tab:prop_constant_bit_upper_a} for an 
overview of 
the hard instance family. The best and second best 
solutions to all 
instances and their associated performances are 
listed in 
Table~\ref{tab:prop_constant_bit_upper_b}. 
\begin{table}
\caption{Hard instance family for \PropRemKnap reading at most $\log k$ advice bits, where $\zeta=(3-2k+\sqrt{4k(k+1)-7})/4$; see the proof of Theorem~\ref{thm:prop_const_bit_lower}.}
\label{tab:prop_constant_bit_upper_a}
\setlength{\tabcolsep}{10.5pt}
\begin{tabular}{l@{\hspace{1.4em}}c@{\hspace{1.6em}}c@{\hspace{1.9em}}Hc@{\hspace{1.9em}}c@{\hspace{2.4em}}c@{\hspace{2.4em}}l}
\toprule
&$x_1$&$x_2$&$x_3$&$\cdots$&$x_k$&$x_{k+1}$&$x_{k+2}=y_j$\\
\cmidrule(r{.9em}){2-8}
$I_1$:    &$\zeta$&$\zeta^2$&$\zeta^2-(1-\zeta)$&$\cdots$&$\zeta^2-(k-2)(1-\zeta)$&$\zeta^2-(k-1)(1-\zeta)+\eps$&None\\
$I_2$:    &$\zeta$&$\zeta^2$&$\zeta^2-(1-\zeta)$&$\cdots$&$\zeta^2-(k-2)(1-\zeta)$&$\zeta^2-(k-1)(1-\zeta)+\eps$&$1-x_2$\\
$\,\vdots$  &$\vdots$&$\vdots$&$\vdots$&$\vdots$&$\vdots$&$\vdots$&$\hphantom{1~}\vdots$\\
$I_{k+1}$:&$\zeta$&$\zeta^2$&$\zeta^2-(1-\zeta)$&$\cdots$&$\zeta^2-(k-2)(1-\zeta)$&$\zeta^2-(k-1)(1-\zeta)+\eps$&$1-x_{k+1}$\\
\bottomrule
\end{tabular}
\end{table}
\begin{table}
\caption{The values of the optimal and best 
suboptimal solutions 
for each instance in the family given in 
Table~\ref{tab:prop_constant_bit_upper_a} and the 
resulting 
competitive performance; see the proof of 
Theorem~\ref{thm:prop_const_bit_lower}.}
\label{tab:prop_constant_bit_upper_b}
\setlength{\tabcolsep}{7.7pt}
\begin{tabular}{l@{\hspace{.9em}}r@{\hspace{4em}}r@{\hspace{7.5em}}l}
\toprule
	&Optimal value&Best suboptimal 
	value\hspace*{-2.5em}&\hspace*{-2em}Best suboptimal 
	performance\\
	\cmidrule(lr){2-4}
	$I_1$:    
	&$x_1=\zeta$&$x_2=\zeta^2$&$\zeta/\zeta^2$\\
	$I_2$:    
	&$x_2+y_2=1$&$x_1=\zeta\phantom{{}^2}$&$1/\zeta\phantom{{}^2}$\\
	$I_3$:    
	&$x_3+y_3=1$&$x_2+y_3=\zeta\phantom{{}^2}$&$1/\zeta\phantom{{}^2}$\\
	$\,\vdots$  
	&$\vdots\hphantom{\,~\zeta}$&$\vdots\hphantom{\,~\zeta^2}$&$\hphantom{1\,}\vdots\hphantom{\,\zeta^2}$\\
	$I_k$:      
	&$x_k+y_k=1$&$x_{k-1}+y_k=\zeta\phantom{{}^2}$&$1/\zeta\phantom{{}^2}$\\
	$I_{k+1}$:&$x_{k+1}+y_{k+1}=1$&$x_k+y_{k+1}=\zeta-\eps\hspace{-1.25em}$&$1/(\zeta-\eps)$\\
\bottomrule
\end{tabular}
\end{table}
\end{proof}

We remark that 
\cref{thm:prop_const_bit_lower} and its analogue for \GenRemKnap instead of 
\PropRemKnap, \cref{thm:general_const_bit_lower}, improve upon 
the best known lower bounds implied 
by Han et al.'s results on the resource 
buffer model~\cite[Thms.~17~and~6]{HKMY2019}. 
In this model, the 
online algorithm may use a knapsack of some 
increased capacity 
$R>1$, but only until the instance ends, at which point it 
has to choose 
from the reserved items a selection that fits a 
knapsack of 
capacity one. 
A resource buffer of some natural size $R$ 
allows us to simulate 
any algorithm using up to $\log R$ advice bits: We 
think of the 
resource buffer as split into $R$ knapsacks of 
capacity $1$, 
allowing us to accommodate the items stored by 
the advice 
algorithm for every possible advice string 
simultaneously. 

Instantiating \cref{thm:prop_const_bit_lower} with $k=2^1$, $k=2^2$, and $k=2^3$, for 
example, we obtain 
the lower bounds 
\begin{align*}
4/(\sqrt{17}-1)\approx{}& 1.2808,\ 
&4/(\sqrt{113}-7)\approx{}& 1.1287,\ &&\text{ and }\ 
&4/(\sqrt{353}-15)\approx{}& 1.0630.
\intertext{
for one, two, and three advice bits, respectively, 
while the lower bounds provided by Han et 
al.~\cite[Thm.~6]{HKMY2019} for a resource buffer of 
size $R=2^k$ are
}
6/5={}& 1.2,
&10/9\approx{}& 1.1111,&&\text{ and }&18/17\approx{}& 1.0588.
\end{align*}
Clearly, the lower bound of 
Theorem~\ref{thm:prop_const_bit_lower} 
tends to $1$ for increasingly large but still 
constant advice. 

With our most surprising result for the 
proportional knapsack 
problem, Theorem~\ref{thm:prop_near_opt_upper}, we will 
prove that 
the true competitive ratio displays the same 
general behavior as 
the lower bound of 
Theorem~\ref{thm:prop_const_bit_lower}: For any 
given $\eps>0$, we can guarantee a competitive 
ratio of $1+\eps$ 
with a constant number of advice bits. 
It is of course also possible to derive more 
specific upper bounds 
for very few advice bits such as the following one.
\begin{theorem}\label{thm:prop_two_bits_upper}
There is a $4/3$-competitive algorithm for 
\PropRemKnap reading two 
advice bits. 
\end{theorem}

\begin{proof}
The algorithm operates in one of four modes, 
depending on the given 
advice. 

We make the following case distinction that 
primarily depends on 
how many items falling into the size interval 
$[1/4,3/4]$---we call 
them \emph{medium} items---appear in the optimal 
solutions.

\subparagraph*{\textit{Strategy One.}}
This strategy is chosen if there is any item 
larger than $3/4$ or 
if there is an optimal solution containing 
either no or only one 
medium item.
The algorithm maintains one maximal medium item 
while packing the 
smaller items greedily.
As an exception, when an item larger than $3/4$ 
appears, it is 
packed and kept to the end, discarding everything 
else. 

This procedure obviously produces 
$(3/4)$-competitive solution if 
there is an item larger than $3/4$. 
Otherwise, the knapsack will be filled optimally if 
all items fit 
into the knapsack together. In the remaining case, 
the greedy 
packing will leave a gap of at most $1/4$ because 
the algorithm 
will never displace a medium item in favor of a 
smaller one, thus 
removing only items smaller than $1/4$.

\subparagraph*{\textit{Strategy Two.}}
This strategy can be chosen if there is an 
optimal solution 
containing at least three medium items or if the 
two minimal medium 
items in the instance have a combined size of 
$3/4$ or more. 
The algorithm maintains as long as possible the 
minimal three 
medium items from the interval among everything 
seen so far. 
If a third medium item does not fit at some point, 
the algorithm 
switches to maintaining only the two minimal 
medium items. 
This clearly yields a filling of at least $3/4$.

\subparagraph*{\textit{Strategy Three.}}
This strategy works if there is an 
optimal solution 
containing two items from the size interval 
$[1/4,1/2]$. 
It maintains two maximal such items, which is 
always possible, and 
packs the smaller items greedily.
This either yields an optimal solution or one 
with a gap of at most 
$3/4$. 

\subparagraph*{\textit{Strategy Four.}}
This strategy is chosen in the remaining case. 
Specifically, we may now assume that every optimal 
solution to the 
given instance contains exactly two medium items, 
one of which has 
a size greater than $1/2$. 

In this case, the algorithm maintains on the one 
hand a minimal 
item and on the other hand a minimal item larger 
than $1/2$. 
Clearly, this fills the knapsack to a total size 
of more than 
$1/4+1/2=3/4$. 
\end{proof}

We now turn to our main result for the 
proportional knapsack 
problem, which complements 
Theorem~\ref{thm:prop_const_bit_lower} 
with an upper bound. 
\Cref{thm:general_log_bits_upper} 
will generalize this result
 to the general version where an item's size may differ from its 
 value, 
albeit with a far more complicated proof. 
To make it as easily understandable as possible, 
we first present here the proof 
for the simple variant, 
which introduces the idea of slots that 
are reserved for items with certain properties. 
This will serve as a useful foundation 
for the proof of the general variant, 
which is also making use of such a slot system, 
although as merely one besides many more components. 

\begin{theorem}\label{thm:prop_near_opt_upper}
For any $\eps>0$, there is a strictly 
$(1+\eps)$-competitive 
algorithm for \PropRemKnap reading a constant 
number of advice 
bits. 
\end{theorem}

\begin{proof}
We describe such an algorithm called 
\textsc{PropPack}; see 
Algorithm~\ref{alg:prop_near_opt_upper} for a pseudo-code 
implementation. We 
begin by describing the advice communicated to 
\textsc{PropPack} 
with a constant number of bits, then explain how 
the algorithm 
operates on this advice, prove that it is correct 
and terminates, 
and finally analyze its competitive ratio.

\subparagraph*{\textit{Notions and Notation.}}
Without loss of generality, we assume that all items have size at most 1 and that $\eps\le 1/2$. We define the constant $K=\lceil\log_{1-\eps/2}\eps/2\rceil$. 

Let an instance with $n$ items be given. 
Denote the items in the order of their appearance in the instance by $1,2,\dots,n$ and denote the size of item $i$ by $\size(i)$. 
We divide the $n$ items into \emph{small} and \emph{big} ones, with $\delta=(1-\eps/2)^K$ serving as the dividing line: 
$C_\klein=\{\,i\mid\size(i)\le\delta\,\}$ and $C_\gross=\{\,i\mid\delta<\size(i)\,\}$.
We further partition the big items into the subclasses 
$
C_k=\{\,i\mid(1-\eps/2)^k<\size(i)\le(1-\eps/2)^{k-1}\,\}\text{
 for 
}k\in\{1,\dots,K\}
$
.
To alleviate the notation, we will often refer to $C_k$ as class $k$ and to $C_\klein$ as class $0$. 
We also use this convention when writing $C(i)$ to indicate the class to which item $i$ belongs: We have  $C(i)\in\{0,\dots,K\}$, with $C(i)=0$ meaning that $i\in C_\klein$ and $C(i)=k\neq0$ meaning that $i\in C_k$.

The oracle chooses an arbitrary but fixed optimal solution $S\subseteq\{1,\dots,n\}$. 
We denote the partition classes that are naturally induced by this solution by
$S_\klein=S\cap C_\klein$, $S_\gross=S\cap C_\gross$, and $S_k= S\cap C_k$ for $k\in\{1,\dots,K\}$.
Let $m=|S_\gross|$ be the number of big items in the optimal solution and denote them by $i_1<\ldots< i_m$ in order of appearance. 

\subparagraph*{\textit{Constant Advice}.}
The oracle communicates to the algorithm a tuple $(b_1,\dots, b_m)$ with the classes of the big items in the chosen optimal solution in order of appearance; that is, we have $b_j=C(i_j)$ for each $j\in\{1,\dots,m\}$. 
We remark that this tuple needs to be encoded in a self-delimiting way. A constant number of bits suffices for this because $b_j$ is bounded by the constant $K$ for every $j\in\{1,\dots,m\}$ and $m$ is bounded by the constant $1/\delta$. 
The latter bound is an immediate consequence of the fact that $\size(S_\gross)\le 1$ and that any big item has a size larger than $\delta$. 

\subparagraph*{\textit{Algorithm Description.}}
The algorithm \textsc{PropPack} proceeds in $m$ phases as follows. 
In every phase, the algorithm opens a new virtual \emph{slot} within the knapsack that can store exactly one item at a time; multiple items in succession are allowed, however. The slot opened in phase $i$ will accommodate items belonging to class $b_i$ exclusively; we say that items from this class \emph{match} slot $i$. Slots are never closed, thus there are exactly $m$ of them in the end. Small items are generally packed in a greedy manner and discarded one by one whenever necessary to pack a big item.  

In the first phase, the algorithm rejects all big items until one of class $b_1$ appears. As soon as this is the case, said item is packed into the first slot, ending the first phase. 

In the second phase, the algorithm opens the 
second slot to pack a matching item, that is, one 
of class $b_2$. It waits for the first item from 
this class that fits into the knapsack alongside 
the item in the first slot. As soon as such an 
item appears, it is packed and the phase ends. In 
the meantime, whenever an item of class $b_1$ 
appears during the second round, the algorithm 
substitutes it for the one stored in the first 
slot if and only if this reduces the size of the 
stored item. 

In general, phase $i$ begins with the opening of slot $i$, which is reserved for items of class~$b_i$. The phase continues until an item appears that both matches the newly opened slot and fits in beside the items currently stored in the previously opened and filled slots without exceeding the capacity. Then this item is packed into the new slot, which ends the phase. 
During the entire phase, the algorithm maintains in all filled slots the smallest matching items seen so far: Whenever the algorithm is presented with a big item  that either it belongs to a class other than $b_i$ or does not fit in alongside the items in the previously opened slots, then the new item replaces a largest item in the matching open slots, unless the new item itself is even larger.
 
The entire time, even after the last phase has terminated, small items are packed greedily and discarded one by one whenever this is necessary to make room for a big item according to the description above. 
Moreover, we may assume that, whenever a new item has been packed into the knapsack, the algorithm sorts the items in the matching open slots in increasing order. This sorting is not necessary for the algorithm to fulfill its duty, but it facilitates the proof by induction below.

\subparagraph*{\textit{Termination of All Phases.}}
We need to show that \textsc{PropPack} does in fact finish all $m$ phases; that is, all $m$ slots will be filled with a matching item without ever exceeding the knapsack capacity. 
Consider the big items of the optimal solution, 
which we denote by 
$u_1<\dots<u_m$ in their order of appearance. 
To ensure the termination of all phases, we prove by induction 
over $i\le m$ 
that, after processing item $u_i$, the first $i$ 
slots store items 
with a total size of $\size(u_1)+\dots+\size(u_i)$ 
or less. 

We may start from $i=0$ as the trivial, if 
degenerate, base case. 
For the induction step, assume the hypothesis for 
$i<m$ and observe 
that no item in a slot is ever replaced by a 
larger one. Therefore, 
the items in the first $i$ slots still have a 
total size of at most 
$\size(u_1)+\dots+\size(u_i)$ when $u_{i+1}$ is 
presented.
 
There are now three possibilities. If slot $i+1$ has 
remained 
closed up to this point, it is now opened and 
filled with 
$u_{i+1}$, which fits in because  
$\size(u_1)+\dots+\size(u_{i+1})\le\size(S_\gross)\le1$.  
Otherwise, slot $i+1$ is already storing an item: If 
said item is 
larger than $\size(u_{i+1})$, then $u_{i+1}$ 
replaces either this 
item or one that is at least as large. During the 
subsequent 
sorting, $u_{i+1}$ is then moved to slot $i+1$ or 
one of the slots 
from $1$ to $i$, which may force some items from 
slots $1$ through 
$i$ into higher slots but never beyond slot 
$i+1$. The third 
possibility is that slot $i+1$ contains an item 
of size at most 
$\size(u_{i+1})$ already. 
We immediately obtain the induction claim for 
$i+1$ in all three 
cases.

\subparagraph*{\textit{Competitive Analysis.}}
We still denote by $S$ the optimal solution that served as the basis for the given advice, by $T$ the final output of the online algorithm \textsc{PropPack} (Algorithm~\ref{alg:prop_near_opt_upper}), and the respective partition classes by $S_\klein$, $S_\gross$, $S_k$ and $T_\klein$, $T_\gross$, and $T_k$.  

Since \textsc{PropPack} opens one slot for each big item in the 
optimal solution $T$ and fills it with an item from the same 
subclass, as proved above, we have $|S_k|=|T_k|$ for every 
$k\in\{1,\dots,K\}$. Moreover, the sizes within a subclass $C_k$ 
vary by a factor of at most $1-\eps/2$; this means that we can 
bound both $\size(S_\gross)$ and $\size(T_\gross)$ from below by 
$L=\sum_{k=1}^K|S_k|(1-\eps/2)^k$ and from above by $L/(1-\eps/2)$. 
We conclude $\size(T_\gross)\ge \size(S_\gross)\cdot(1-\eps/2)$. 

Furthermore, since small items are packed greedily and only discarded one by one whenever necessary to make room for the big items, we will not lose much from their side either. If the presented small items have a total size of at most $1-L/(1-\eps/2)$, none is ever discarded. In this case, we have $\size(T_\klein)\ge \size(S_\klein)$ and thus immediately $\size(T)\ge \size(S)\cdot(1-\eps/2)$. If small items are discarded, however, the worst case is the following type of instance:  It starts with only small items of the largest possible size $\delta$, some of which are then discarded to accommodate big items with sizes right at the upper limit for the classes indicated by the advice, leaving a gap of almost $\delta$, follows up with slightly smaller big items that are in the optimal solution and would not have lead to any discarded small items, and finally presents big items at the lower end of the size span, replacing all previously packed big items. 

Even in this worst case, the algorithm remains $(1-\eps/2)$-competitive on the big items and detracting the largest possible loss of $\delta$ on the small items yields $\size(T)\ge\size(S)\cdot(1-\eps/2)-\delta$.
By the definition of $\delta$ and $K$ and due to the simple fact that $\size(S)$ is at most $1$, we have $\delta=(1-\eps/2)^K\le \eps/2\le \size(S)\cdot \eps/2$. This implies $\size(T)/\size(S)\ge 1-\eps$, as desired. 
\end{proof}

\begin{algorithm}[H]
	\caption{{}\vphantom{$0^{0^0}$}\textsc{PropPack}}
	\label{alg:prop_near_opt_upper}
	\medskip
	\textbf{Parameter: } Any  $\eps\in(0,1/2]$.

	\medskip
	\textbf{Online Input:} A sequence $I=(1,\dots,n)$ of $n$ items with sizes $(s_1,\dots,s_n)$.

	\medskip
	\textbf{Online Output:} 
	A $(1+\eps)$-competitive packing $T=T_\klein\cup 
	T_\gross$.

	\medskip
	\textbf{Advice: } The sequence $B=(b_1,\dots,b_m)$, 
	where $m$ 
	is the number of big items in a fixed 
	optimal solution and 
	$b_j$ is the class of the $j$th big item 
	appearing in it.

	\medskip
	\textbf{Algorithm:}
\medskip
	\begin{algorithmic}[1]
		\State $k \gets \Next(B)$ \Comment{Initialize $k$ to class of first big item to be packed.}
		\State $T_\klein \gets \emptyset$ \Comment{Initialize $T_\klein$, set of packed small items, to the empty set.}
		\State $T_\gross \gets \emptyset$ \Comment{Initialize $T_\gross$, set of packed big items, to the empty set.}
		\For{$i$ \In $I$} \Comment{For each new item in order of appearance do the following:}
			\If{$C(i)=0$} \Comment{If the new item is small, then check if it \dots}
				\If{$\size(T_\klein\cup T_\gross\cup 
				\{i\})\le 1$}\Comment{\dots{}\;fits in beside everything currently packed; \dots}
					\State $T_\klein\gets T_\klein\cup 
					\{i\}$ \Comment{\dots{}\;if it does, then pack it.}
				\EndIf
			\ElsIf{$C(i)=k$ \And 
			$\size(T_\gross\cup\{i\})\le1$}   \Comment{If big item of advised class can be fit in, \dots}
				\While{$s(T_\klein\cup 
				T_\gross\cup\{i\})>1$} \Comment{\dots{}\;then, until it actually fits, \dots}
					\State $\Pop(T_\klein)$ \Comment{\dots{}\;greedily discard small items one by one.}
				\EndWhile
				\State $T_\gross \gets T_\gross\cup\{i\}$ \Comment{Now that it actually fits, pack the new big item.}
				\label{line:fillnewslot}
				\State $k \gets \Next(B)$ \Comment{Update $k$ to class of the next item advised to be packed.}
			\Else \Comment{Among big new item and kept ones of same class, remove largest one and \dots}
				\State $T_\gross\gets 
				(T_\gross\cup\{i\})\setminus 
				\arg\max\{s(j)\mid j\in \{i\}\cup 
				(C_{C(i)}\cap 
				T_\gross)\}$ \Comment{\dots{}\;pack the rest.}
			\EndIf
		\EndFor
		\State \Return $T_\klein\cup T_\gross$ \Comment{Return the current solution after processing the entire input.}
\medskip
	\end{algorithmic}
\end{algorithm}

\section{Results for General Removable Knapsack}
First, we note that all lower bounds for the 
proportional removable 
knapsack problem carry over to the general 
removable knapsack 
problem, in particular \cref{thm:prop_opt_lower}. 

Iwama and Zhang~\cite{IZ2010} have shown that the 
competitive ratio of \GenRemKnap is unbounded 
without advice.  
This can be seen 
using an interactive instance that starts with 
an item $(1,1)$ and 
then presents items $(\eps^2,\eps)$ repeatedly, up to 
$1/\eps^2$ 
times, until one is packed, at which point the 
instance ends. 

The following two theorems show \GenRemKnap's
competitivity for one 
advice bit to be 
exactly $2$.

The existence of a $2$-competitive algorithm already follows from 
a result by Han et 
al.~{\cite[Thm.~9]{HKM2015}}, who proved the statement 
even for a single random bit instead of an advice bit. 
We prove \cref{thm:general_one_bit_upper} by describing a 
concrete advice algorithm. 

\begin{theorem}\label{thm:general_one_bit_upper}
There is a $2$-competitive algorithm for \GenRemKnap 
reading only a single advice bit. 
\end{theorem} 

\begin{proof}
The single advice bit indicates whether the 
instance contains an 
item worth at least half of the optimal solution 
value. 
If so, the algorithm greedily packs the most 
valuable item.
Otherwise, it packs in a yield-greedy manner while 
ignoring any 
items larger than $1/2$. If some optimal solution 
contains some 
item larger than $1/2$, then the rest of this 
solution is smaller 
than $1/2$ and worth at least half of the optimal 
solution value. 
The yield-greedy algorithm achieves at least 
this value: Either it 
packs all items up to size $1/2$ or, when having 
to discard such an 
item, leaves a gap smaller than $1/2$ while the 
rest is filled with 
the best possible yield.
\end{proof}

We now provide the matching lower bound for \cref{thm:general_one_bit_upper}.

\begin{theorem}\label{thm:general_one_bit_lower}
No algorithm for \GenRemKnap reading only 
a single advice 
bit can 
have a competitive ratio better than $2$.
\end{theorem} 

\begin{proof}
Fix an arbitrary positive $\eps<1/2$ such that 
$k=1/\eps$ is an 
integer. 
We describe an adversarial instance family on 
which no algorithm 
with a single advice bit can achieve a 
competitive ratio better 
than $2(1-\eps)$. 

The instance will present some subset of the 
following items:
\begin{align*}
&&&&&x_i\quad=&(1-i\eps^3,\ &2-i\eps)&&&&&&\text{ for 
}i\in\{0,1,\dots,k\},&&&&&&&&\\
&&&&&x_i'\quad=&(i\eps^3,\ &2-i\eps+\eps)&&&&&&\text{ for 
}i\in\{1,2,\dots,k\},\text{ and }&&&&&&&&\\
&&&&&y_j\quad=&(\eps,\ &4\eps)&&&&&&\text{ for 
}j\in\{1,2,\dots,k\}.&&&&&&&& 
\end{align*}

The exact subset and order of presentation 
depends on the operating 
advice algorithm and will be explained below.

We begin by making some simple observations. 
The items 
$x_0,x_1,\dots,x_k$ 
decrease 
in both size and value. The smallest of them has 
still size 
$1-\eps^2$, thus it is impossible to fit two of 
them into the 
knapsack 
together. The same is true for combining any item 
$x_i$ with any 
item $y_j$, together the exceed the knapsack 
capacity. 
Finally, an item $x_i'$ fits together with $x_i$ 
perfectly, but not 
with any of the items $x_0,x_1,\dots,x_{i-1}$. 

We arrange the items along the axes of a 
$(k+1)\times k$ grid as 
shown in \cref{fig:grid}.
Every instance in the hard instance family can 
be represented by a 
directed path in this grid that starts at the 
top left corner, only 
moves down or 
right from there, and stops at the latest when 
reaching either the 
bottom or right border of the grid. 
\Cref{fig:grid} shows one possible path.

We can read a given path as follows.
All instances start by presenting the two items 
$x_0$ and $y_1$ 
corresponding to the corner $(x_0,y_1)$. 
When a new coordinate along the $x$-axis or 
$y$-axis is reached, 
the corresponding item is presented. 
Thus, the path moving down means presenting the 
next item in the 
sequence $x_0,x_1,\dots,x_k$, and analogously for 
moving to the right and the sequence 
$y_1,y_2,\dots,y_k$. 
If the path ends at $(x_i,y_j)$ with $i<k$ and 
$j<k$---that is, 
before reaching the grid's bottom or right 
border---then the 
instance 
concludes with $x_k'$ as an additional final item. 
If the path reaches the bottom or right end, the 
instance ends 
without such an additional item.  

\begin{figure}[ht]
\centering
\begin{tikzpicture}[xscale=1.1,yscale=-1.1,
every node/.style={font=\small}]
\tikzset{
    dot/.style 2 args={fill, circle, inner sep=1pt, 
    label={#1:\scriptsize #2}}
}
\node at (-4,0) {};

\draw[very thin,dotted] (0.01,0.01) grid (4.99,4.99);

\draw[very thin] (0,0) grid (2.18,2.18);

\draw[very thin] (0,2.82) grid (2.18,3.18);
\draw[very thin] (3.82,2.82) grid (5,3.18);
\draw[very thin] (2.82,0) grid (3.18,2.18);
\draw[very thin] (2.82,3.82) grid (3.18,5);

\draw[very thin] (2.82,2.82) grid (3.18,3.18);

\draw[very thin] (3.82,0) grid (5,2.18);
\draw[very thin] (0,3.82) grid (2.18,5);

\draw[very thin] (3.82,3.82) grid (5,5);

\foreach \j [count=\enum from 0] in {1,2,3,j,k-1,k} \node 
at 
(\enum,-.3) 
{$y_{\j}$};
\foreach \j [count=\enum from 0] in {=,=,\cdots,\cdots,=,=} 
\node 
at 
(\enum+.5,-.3) 
{$\j$};
\node at (6.2,-.3) {$(\eps,4\eps)$};

\foreach \i [count=\enum from 0] in 
{
(1,2)=x_{0\hphantom{{}-1}},
(1-\eps^3,2-\eps)=x_{1\hphantom{{}-1}},
(1-2\eps^3,2-2\eps)=x_{2\hphantom{{}-1}},
(1-i\eps^3,2-i\eps)=x_{i\hphantom{{}-1}},
(1-\eps^2+\eps^3,1+\eps)=x_{k-1},
(1-\eps^2,1)=x_{k\hphantom{{}-1}}} 
\node[anchor=east] at (-.11,\enum) 
{$\i$};

\foreach \i in {2,3} \node at 
(-1,\i+.5) 
{$\vdots$};

\draw[ultra thick,marrow=triangle 45] (0.01,0) -- (1,0);
\draw[ultra thick,marrow=triangle 45] (1,0) -- (1,1);
\draw[ultra thick,marrow=triangle 45] (1,1) -- (1,2);
\draw[ultra thick,marrow=triangle 45] (1,2) -- (2,2);
\draw[ultra thick] (2,2) -- (2.18,2);
\draw[ultra thick,dotted] (2.18,2) -- (2.82,2);
\draw[ultra thick] (2.82,2) -- (3,2);
\draw[ultra thick] (3,2) -- (3,2.18);
\draw[ultra thick,dotted] (3,2.18) -- (3,2.82);
\draw[ultra thick] (3,2.82) -- (3,3.18);
\draw[ultra thick,dotted] (3,3.18) -- (3,3.82);
\draw[ultra thick] (3,3.82) -- (3,4);
\draw[ultra thick,marrow=triangle 45] (3,4) -- (3,4.99);

\draw[thick] (0,0) -- (2.18,0);
\draw[thick,dotted] (2.18,0) -- (2.82,0);
\draw[thick] (2.82,0) -- (3.18,0);
\draw[thick,dotted] (3.18,0) -- (3.82,0);
\draw[thick] (3.82,0) -- (5,0);

\draw[thick] (0,0) -- (0,2.18);
\draw[thick,dotted] (0,2.18) -- (0,2.82);
\draw[thick] (0,2.82) -- (0,3.18);
\draw[thick,dotted] (0,3.18) -- (0,3.82);
\draw[thick] (0,3.82) -- (0,5);

\draw[thick] (0,5) -- (2.18,5);
\draw[thick,dotted] (2.18,5) -- (2.82,5);
\draw[thick] (2.82,5) -- (3.18,5);
\draw[thick,dotted] (3.18,5) -- (3.82,5);
\draw[thick] (3.82,5) -- (5,5);

\draw[thick] (5,0) -- (5,2.18);
\draw[thick,dotted] (5,2.18) -- (5,2.82);
\draw[thick] (5,2.82) -- (5,3.18);
\draw[thick,dotted] (5,3.18) -- (5,3.82);
\draw[thick] (5,3.82) -- (5,5);

\node[fill, circle, minimum size=2.5mm, inner sep=0mm] 
at (0,0) {};
\node[fill, circle, minimum size=2.5mm, inner sep=0mm] 
at (3,5) {};

\end{tikzpicture}
\caption{The grid described in the proof of 
\cref{thm:general_one_bit_lower}. The directed path 
represents the instance that presents the items 
$x_0,y_1,y_2,x_1,x_2,y_3,y_4,\dots,y_j,x_3,x_4,\dots,x_k$ in 
this 
order.}
\label{fig:grid}
\end{figure}

We now describe how the path is determined by 
the actions of the 
advice algorithm.
Note that an algorithm with a single advice bit 
can be interpreted 
as two deterministic algorithms with the advice 
bit determining 
which algorithm is executed on any given instance. 
From $(x_i,y_j)$ with $i<k$ and $j<k$, the path 
continues 
as follows. 
\begin{description}
\item[Case 1.] If both algorithms have discarded 
$y_j$, then it 
goes to the 
right. 
\item[Case 2.] 
If one algorithm has $y_j$ in its reserve and 
the other $x_i$, 
the path continues downward. 
\item[Case 3.] 
If one algorithm has $y_j$ in its reserve, but 
neither 
has kept $x_i$, then the path stops.  
\end{description}

We make two observations. 
First, at 
any point $(x_i,y_j)$, the two algorithms may have 
$y_j$ in their 
reserve, but no other item from $y_1,y_2,\dots,y_k$. 
This is 
because the path is moving right---and thus 
the instance presenting an item in this 
sequence---only after 
its predecessor has been discarded by both 
algorithms. 

Second, when arriving at $(x_i,y_j)$, one of the 
algorithms may 
still have $x_{i-1}$ in its reserve, but neither 
algorithm will 
have kept 
any of the previous items $x_0,\dots,x_{i-2}$. This 
is because the 
path 
was able to move down to the current 
$x$-coordinate only if one 
algorithm has had $x_{i-1}$ in its reserve and 
the other some item  
$y_{j'}$, which excludes any third item from 
$x_0,\dots,x_{i-2}$.

We consider the three ways that a path 
representing an instance 
can end: at the right border, at the bottom edge, or 
anywhere else. 

\begin{description}
\item[Case 1.] 
If the path ends at the right border, the optimal 
solution consists 
of the entire sequence $y_1,y_2,\dots,y_k$ with 
total size 
$k\cdot\eps=1$ and total value $k\cdot4\eps=4$. 
The two deterministic algorithms, in contrast, have 
discarded all 
items in this optimal solution but possibly 
$y_k$. Among $y_k$ and
$x_0,x_1,\dots,x_k$, which might have been presented, 
no two fit 
into a 
knapsack of capacity 1 together, thus the best 
feasible 
solution for the advice algorithm 
is taking the single most valuable item $x_0$ 
with a value of $2$. 
The competitive performance of the algorithm is 
$4/2=2$ in 
this case.

\item[Case 2.]
If the path reaches the bottom, this means that 
the last move was 
downward; thus one algorithm has kept $x_{k-1}$ in 
its reserve and 
the other $y_j$. Therefore, only the three 
singleton solutions 
consisting of $x_{k-1}$, $x_k$, and $y_j$, 
respectively, are 
attainable by the advice algorithm. Among these, the 
first option 
is the best with a value of $1+\eps$. The optimum 
would have been 
to 
keep $x_0$ with a value of $2$, however, resulting 
in a competitive 
performance of $2/(1+\eps)$ or worse.

\item[Case 3.]
Finally, the path may stop at some point $(x_i,y_j)$ 
with $i<k$ and 
$j<k$. In this 
case, 
both algorithms have discarded $x_i$ and, as 
observed in the 
beginning, all previously presented items except 
for $x_{i-1}$ and 
$y_j$. 
Thus, when the final item $x_i'$ is 
presented, the best option for the advice algorithm 
is the 
singleton solution with the item $x_{i-1}$ of 
value $2-(i-1)\eps$.  
The optimal solution, in contrast, can combine the 
two items $x_i$ 
and $x_i'$ of complementing sizes and a total 
value of 
$2(2-i\eps)+\eps=2(2-(i-1)\eps)-\eps$. 
The resulting competitive performance is 
$2-\eps/(2-(i-1)\eps)\le 
2-\eps/(2-(k-1)\eps)\le 2-\eps$.
\end{description}

Overall, the advice algorithm's competitive ratio 
cannot be better 
than 
\[\min\big\{2,\frac2{1+\eps},2-\eps\big\}\ge 2-2\eps,\]
which tends to $2$ for decreasing $\eps$, thus 
proving the theorem.
\end{proof}

Advice being more powerful than randomness, 
\cref{thm:general_one_bit_lower} also closes 
the remaining gap for barely random algorithms 
by lifting the 
previously best known lower bound by Han et 
al.~\cite[Thm.~12]{HKM2015} from $1+1/e\approx1.367$ to 
$2$. 

The analysis of the hard instance presented in 
the proof of 
\cref{thm:general_one_bit_lower} could be adapted to the 
case of more than one advice bit. 
Two advice bits would mean that the oracle can provide to the algorithm one out of four advice strings instead of the two advice strings possible with one bit.
We may also consider an intermediate advice algorithm limited to three advice 
strings, which corresponds to $\log 3$ advice bits. 
Such an algorithm cannot 
achieve 
a competitive ratio 
better than $2/\Phi=4/(1+\sqrt5)\approx 1.2361$ since it 
is forced to use 
one of the three advice strings to keep the most 
valuable item 
$x_0$, one to pack the 
items $y_1,\dots,y_k$ yield-greedily, and one to keep 
the $x_j$ 
with the value 
$\Phi\approx 1.618$, resulting in a competitive ratio 
of 
$2/\Phi=2\Phi/(1+\Phi)$. 

This approach deteriorates too quickly, however. 
Adapting 
\cref{thm:prop_const_bit_lower} to the 
case of \PropRemKnap is the better choice; this 
results in \cref{thm:general_const_bit_lower}, which 
already yields a better bound in 
the case of three advice strings, that is, for 
$k=3$. For increased accessibility, we first provide an instantiation 
for the case of a single advice bit, namely $k=2$. 

\begin{theorem}\label{thm:general_one_bit_lower_obsolete}
No algorithm for \GenRemKnap reading only 
a single advice 
bit can 
have a better competitive ratio than 
$(1+\sqrt{3})/2\approx 
1.36603$.
\end{theorem}

\begin{proof}
The proof is similar to the one of 
Theorem~\ref{thm:prop_one_bit_lower} but takes 
advantage of the 
fact that the sizes and values of the items can 
be chosen 
independently.
Let $\eta = (\sqrt{3}-1)/2\approx 0.366025$ and let an 
algorithm 
for general removable knapsack reading only a 
single advice bit 
be given. 
Consider the three instances displayed in 
Table~\ref{tab:general_one_bit_upper}, which all start 
with 
the same three items.  
\begin{table}[h]
\caption{Hard instance for an algorithm for \GenRemKnap 
reading one 
advice bit, where $\eta=(\sqrt3-1)/2$ satisfies 
$1/(2\eta)=1+\eta=(1+\sqrt3)/2\approx1.36603$; see the proof of \cref{thm:general_one_bit_lower_obsolete}.}
\label{tab:general_one_bit_upper}
\begin{tabular}{c@{\hspace{1.9em}}c@{\hspace{1.5em}}c@{\hspace{1.5em}}c@{\hspace{1.5em}}c@{\hspace{1.4em}}c@{\hspace{2.2em}}cccc}
\toprule
	&$x_1$&$x_2$&$x_3$&$x_4$&$x_4'$&optimal&suboptimal&ratio\\
	\cmidrule(r{2.1em}){2-6}\cmidrule(r){7-9}
	$I_1$:&$(1,1)$&$(0.9,2\eta)$&$(0.8,\eta)$&&&$1$&$2\eta$&$1/2\eta$\\
	$I_2$:&$(1,1)$&$(0.9,2\eta)$&$(0.8,\eta)$&$(0.1,1-\eta)$&&$1+\eta$&$1$&$1+\eta$\\
	$I_3$:&$(1,1)$&$(0.9,2\eta)$&$(0.8,\eta)$&&$(0.2,1)$&$1+\eta$&$1$&$1+\eta$\\
\bottomrule
\end{tabular}
\end{table}

For each $i\in\{1,2,3\}$, the instance $I_i$ has a 
unique optimal 
solution; it contains $x_i$ plus the last fourth 
item if it exists. 
Table~\ref{tab:general_one_bit_upper} shows the total 
size for each 
of these optimal solutions and the second best 
solution. 

Because the sizes of any two of these three 
items $x_1$, $x_2$, and 
$x_3$ sum up to over 1, the algorithm can have at 
most one of them 
in the knapsack after being offered $x_3$. 
Moreover, since we have only one advice bit but 
three instances, 
there are two instances for which the algorithm 
has 
packed the same item right before the potential 
presentation of the 
fourth 
item. 
This implies that the algorithm is suboptimal 
for at least one 
instance, thus its competitive ratio cannot be 
better than 
$1/(2\eta)=1+\eta=(1+\sqrt{3})/2$. 
\end{proof}

\begin{theorem}\label{thm:general_const_bit_lower}
Let an arbitrary integer $k>1$ be given. No 
algorithm for 
\GenRemKnap 
reading at most $\log k$ advice bits can achieve 
a better 
competitive ratio than $1/2+\sqrt{1/4+1/k}$.
\end{theorem}

\begin{table}
\caption{A hard instance family for \GenRemKnap 
reading at most $\log k$ 
advice bits; see the proof of 
Theorem~\ref{thm:general_const_bit_lower}. Only the 
values of the 
items are 
given, using $\xi=(3-2k+\sqrt{4k(k+1)-7})/4$; the sizes 
can be 
chosen arbitrarily 
satisfying $1\ge 
\size_1>\size_2>\dots>\size_k>\size_{k+1}>1/2$ and 
$\size_{k+2}=1-\size_j$ in instance $I_j$.}
\label{tab:general_constant_bit_upper_a}
\setlength{\tabcolsep}{10.5pt}
\begin{tabular}{l@{\hspace{1.9em}}c@{\hspace{2.2em}}c@{\hspace{2.2em}}Hc@{\hspace{2.1em}}c@{\hspace{2.1em}}c@{\hspace{2.1em}}l}
\toprule
&$\val_1$&$\val_2$&$\val_3$&$\cdots$&$\val_k$&$\val_{k+1}$&~~$\val_{k+2}$\\
\cmidrule(r){2-8}
$I_1$:    
&$1$&$1/\xi$&$1/\xi-(\xi-1)$&$\cdots$&$1/\xi-(k-2)(\xi-1)$&$1/\xi-(k-1)(\xi-1)$&0\\
$I_2$:    
&$1$&$1/\xi$&$1/\xi-(\xi-1)$&$\cdots$&$1/\xi-(k-2)(\xi-1)$&$1/\xi-(k-1)(\xi-1)$&$\xi-\val_2$\\
$\,\vdots$  
&$\vdots$&$\vdots$&$\vdots$&$\vdots$&$\vdots$&$\vdots$&$\hphantom{1~}\vdots$\\
$I_{k+1}$:&$1$&$1/\xi$&$1/\xi-(\xi-1)$&$\cdots$&$1/\xi-(k-2)(\xi-1)$&$1/\xi-(k-1)(\xi-1)$&$\xi-\val_{k+1}$\\
\bottomrule
\end{tabular}
\end{table}

\begin{table}
\caption{The values of the optimal and best 
suboptimal solutions 
for each instance in the family given in 
Table~\ref{tab:general_constant_bit_upper_a} and the 
resulting 
competitive performance; see the proof of 
Theorem~\ref{thm:general_const_bit_lower}.}
\label{tab:general_constant_bit_upper_b}
\setlength{\tabcolsep}{7.7pt}
\begin{tabular}{lr@{\hspace{3.4em}}r@{\hspace{6.9em}}l}
\toprule
	&Optimal value&Best suboptimal 
	value\hspace*{-2.5em}&\hspace*{-2em}Best suboptimal 
	performance\\
	\cmidrule(lr){2-4}
	$I_1$:    &$\val_1=1$&$\val_2=1/\xi$&$\xi$\\
	$I_2$:    
	&$\val_2+\val_2'=\xi$&$\val_1=\val_3+\val_2'=1$&$\xi$\\
	$I_3$:    
	&$\val_3+\val_3'=\xi$&$\val_1=\val_4+\val_3'=1$&$\xi$\\ 
	$\,\vdots$ 
	&$\vdots\hphantom{\,~\xi}$&$\vdots\hphantom{\,~\xi}$&$\vdots\hphantom{\,\xi}$\\
	$I_k$:      
	&$\val_k+\val_k'=\xi$&$\val_1=\val_k+\val_{k-1}'=1$&$\xi$\\
	$I_{k+1}$:&$\val_{k+1}+\val_{k+1}'=\xi$&$\val_{k+1}+\val_k'=1$&$\xi$\\
\bottomrule
\end{tabular}
\end{table}

\begin{proof}
Let an arbitrary integer $k>1$ be given and 
define $\xi$ as the 
unique positive solution of $1/\xi=k(\xi-1)$, 
namely 
$\xi=1/2+\sqrt{1/4+1/k}$.
Consider $k+1$ instances that all start with the 
same $k+1$ items 
of the following, decreasing values: $\val_1=1$ and 
$\val_i=1/\xi-(i-2)(\xi-1)$ for every 
$i\in\{2,\dots,k+1\}$. Note 
that $\val_{k+1}=\xi-1$.  The sizes can be chosen 
arbitrarily 
satisfying $1\ge 
\size_1>\size_2>\dots>\size_k>\size_{k+1}>1/2$. 

The instance $I_1$ ends immediately after these 
common items, 
whereas the instance $I_i$, for $i\in\{2,\dots,k+1\}$, 
presents a 
complement to item $i$, namely an item of size 
$\size_i'=1-\size_i$ 
and value $\val_i'=\xi-\val_i$ as the final one.
There is a unique optimal solution for each 
instance: For $I_1$, it 
is to pack the first item, which has value 
$\val_1=1$. For $I_i$ 
with $i>1$, it is to pack item $i$ and its 
complement, which sum up 
to the optimal solution value $\xi$. 
Since there are only $\log k$ advice bits 
available to handle the 
$k+1$ instances, at least two instances $I_i$ and 
$I_j$ with $i<j$ 
are processed with the same advice string and 
thus by the same 
deterministic algorithm. Consider this algorithm 
and the moment 
after seeing and taking decisions on the first 
$k+1$ items. 
It is impossible for the algorithm to have more 
than one of these 
common items packed since all of these items are 
larger than $1/2$. 
Now, if item $i$ is packed at the considered 
moment, the algorithm 
will perform suboptimally on instance $I_j$. 
Analogously, if item 
$j$ is packed, the performance on instance $I_i$ 
is suboptimal. 

Now if suffices to check that the best 
suboptimal solution has a 
value of at most $1/\xi$ for $I_1$ and at most 
$1$ for the other 
instances. This leads to a performance ratio of 
$\xi$, thus 
proving the theorem. 
See Table~\ref{tab:general_constant_bit_upper_a} for an 
overview of 
the 
hard instance family. The best and second best 
solutions to all 
instances and their associated performances are 
listed in 
Table~\ref{tab:general_constant_bit_upper_b}. 
\end{proof}

We point out again that \cref{thm:general_const_bit_lower} 
slightly improves 
over the lower bounds known from the resource buffer 
model by Han 
et al.~\cite[Thm.~6]{HKMY2019}. 
Specifically, we can choose $k=2^1$, $k=2^2$, and $k=2^3$, 
for example, 
we obtain the lower bounds 
\begin{align*}
\frac{1+\sqrt3}2\approx{}& 1.3660,\quad
&\frac{1+\sqrt2}2\approx{}& 1.2071,\quad&&\text{ and }\quad
&\frac{1+\sqrt{3/2}}2\approx{}& 1.1124.
\intertext{
for one, two, and three advice bits, respectively. 
The corresponding lower bounds by Han et al.\ for a resource buffer of 
size $R=2^k$ are
}
4/3\approx{}& 1.3333,
&6/5={}& 1.2,&&\text{ and }&10/9\approx{}& 1.1111.
\end{align*}

We now move on to the core result of this paper, 
proving that a constant amount of 
advice bits is 
sufficient to reach a near-optimal competitive 
ratio not only for 
the proportional but even for the 
general removable knapsack problem. 
This will complete the picture of the global 
advice behavior of the 
online knapsack problem with removability 
outlined in  
\cref{fig:prop_schematics_combined}.

We first point out that algorithm \textsc{PropPack} 
(Algorithm~\ref{alg:prop_near_opt_upper}) generally does not work 
on 
instances where the value of an item can vary 
independently of its 
size, as seen by the following counterexample: 
Assume that $1/2$ lies in the interior of some size class and choose an $\eps>0$ such that $1/2+\eps$ and $1/2-\eps$ are still in the same class. Present two items $(1/2+\eps,2)$ and $(1/2,1)$. 
If the algorithm picks the first one, $(1/2,2)$ is presented as the last item; otherwise, the instance ends with the item $(1/2-\eps,1)$. In both cases, the algorithm achieves a total value of $2$, whereas the optimum is $3$. The advice does not help us to distinguish the two cases, it only tells us that the optimal solution contains two items from the class. 
Clearly, we have to adapt the algorithm to take the value of the items into account somehow. A major obstacle is that the online algorithm has no bound on the values of appearing items, thus the algorithm has no way of reconstructing the constantly many value classes used by the oracle just from the parameter $\eps$. 

Moreover, the proof of Theorem~\ref{thm:prop_near_opt_upper} cannot be adapted for the general case in any simple way. 
Using only classes based on size, it is impossible for the algorithm to know, when maintaining an item in slot, how to balance minimizing the size against maximizing the value. 
On the one hand, if the size is not minimized, then the excess size may prevent other slots from being filled.  
On the other hand, not maximizing the values, the algorithm may incur an arbitrarily high loss because the potential values of items cannot be bounded. 
This is also the reason why simple value classes do not work either. 
The algorithm does not know the maximal value occurring in the instance until it has ended and can therefore not use it as a reference point, in contrast to the size classes that can be chosen relative to the known knapsack capacity. 

A first step toward solving these issues is the definition of dynamic value classes 
that are anchored to both the value of the first item 
appearing in the instance and to the optimal 
solution value. The latter is of course also unknown to the algorithm until the instance ends. 
However, we are able to define our classes with some additional properties that 
enable our algorithm to compute at any point useful provisional bounds on the optimal solution value. 
These bounds will either turn out to be valid or the algorithm is able to notice 
that they are off just in time to adjust and take a fresh start before having lost too much due to bad decisions. 
The adversary may foil the algorithm over and over, forcing it to abandon its plans and adjust the bounds 
arbitrarily often. 

To properly deal with these repeated resets, we develop a \emph{level system}. 
One major challenge is to square the level system with some sort of \emph{slot system} 
as used by the algorithm for the proportional case. 
We manage to do this by introducing the concept of a \emph{virtual algorithm}, 
which has the special, even though only imagined, capability of keeping one 
item in a \emph{splitting slot} and use arbitrary fractions of the item stored in it. 
We then describe an actual algorithm that tries to equal the idealized performance 
of the virtual algorithm without making use of the splitting slot. 
While it cannot quite achieve this, it will fare well enough in the end. 
Having one algorithm emulating another, we are going 
to prove the claimed competitivity in two stages, 
first for the virtual version and then for the actual algorithm. 

This split analysis presents several further challenges, for example, a desynchronization 
of the current phase and the number of slots filled by the algorithm, which coincided in the proportional case.
The necessary adaptations entail a number of further challenges, for example a judicious handling of the \emph{paltry} items, which are worth almost nothing individually, yet may be too numerous to neglect. 
In fact, the algorithm will need to partition the items not only by their value but simultaneously by their size as well. 

Overcoming these and a few other obstacles, we are able to prove our final theorem. 

\begin{theorem}\label{thm:general_log_bits_upper}
For any $\eps>0$, there is a strictly 
$(1+\eps)$-competitive 
algorithm for \GenRemKnap reading a constant number 
of advice bits. 
\end{theorem}

\begin{proof}
Proving this theorem requires far more effort 
than what was necessary for its proportional counterpart \cref{thm:prop_near_opt_upper}, 
where an item's value is always identical to its size. 
Having explained already why any straightforward adaption of the 
substantially simpler approach for the proportional variant is 
impossible, 
we now provide a high-level outline of the proof. 
Then, we introduce 
some notions and notation necessary to describe 
what advice the oracle is communicating to the 
algorithm, and show that encoding this advice is possible with a 
constant number of bits. 
Then, we describe the algorithm and how it is using the advice; 
see Algorithm~\ref{alg:general_log_bits_upper} for an 
implementation in pseudo-code. 
Finally, we use two proofs by induction to show that the 
algorithm's 
online output successfully maintains some properties, which help us 
to conclude the proof by bounding the competitive ratio.

\paragraph*{\textit{Outline}}
As announced, we first provide a high-level outline of the 
workings of the advice algorithm.

The oracle chooses an arbitrary optimal solution $S$ and, 
based on its value 
$\val(S)$, partitions 
the items of the instance into \emph{precious} 
and \emph{paltry} 
ones; 
the paltry ones are those worth less than 
$\eps_\paltry\cdot\val(S)$ for 
a suitable 
$\eps_\paltry>0$. 
The precious items are further split into 
finitely many classes 
such that the item 
values within any class are at most some factor 
$1-\eps_\spread$ apart 
from each other. 

The advice will encode exactly the classes 
of the precious items 
from the optimal 
solution $S$ in their order 
of appearance; the goal is to ensure that the 
algorithm packs just 
as many precious 
items from each class as the optimal solution 
does, thus achieving a competitive factor of  
$1/(1-\eps_\spread)$ on these items. 
Assuming that the algorithm knows the exact 
value ranges of each 
class, it can achieve 
this using a system of \emph{slots}, which 
will be filled with precious items in two stages: first just  
\emph{virtually}---assuming the algorithm were able to do certain 
things that are in fact impossible---and then also \emph{actually} 
at some point. Each virtual filling of a slot starts a new 
\emph{phase} of the algorithm. 

However, the algorithm knows only the target value but 
nothing about the 
size of the items 
belonging into each slot; it is necessary to 
prove that despite 
this, the algorithm is 
able to fill the slots in the right order without 
blocking important 
items yet to come. 

And there is another problem: It is even 
impossible for the oracle 
to communicate to the 
algorithm the exact value 
range of each class with a constant amount of advice since 
there is no bound on 
the potential values 
occurring in an instance.
Instead, the value ranges will be described in 
relation to the value 
of the first item of 
the instance, and merely modulo some constant 
factor. 
The algorithm will then operate under the 
assumption that the first 
item is a precious 
one, in which case all of the above will work out. 
Since the algorithm cannot know for sure which 
items are precious, 
it divides them into 
\emph{presumably precious} and \emph{provenly 
paltry} ones according to some computations based 
on the advice and the instance seen so far. 
If the algorithm's assumption is mistaken, it is able 
to recognize 
this just in time by 
continually comparing the best solution 
realizable with the already presented items---whether they have 
been 
accepted or not---to a rather intricate estimate for the 
optimal solution 
value $\val(S)$. 
Once 
the algorithm discovers its mistake, it resets 
with a revised set 
of assumptions 
on the value ranges; 
we say that the algorithm \emph{levels up}.   
This is done such that the algorithm can go 
through arbitrarily 
many levels, resetting and taking a fresh start as 
often as 
necessary without incurring more than a 
negligible value loss. 

The algorithm also needs to take care of the 
paltry items, which 
might constitute a 
considerable part of the optimal solution if 
there are sufficiently 
many of them. 
The algorithm is packing the paltry items in a 
somewhat inhibited greedy manner that 
optimizes the value-to-size ratio, which we also 
call \emph{yield}. 
The volume taken up by the paltry items 
will be restricted sufficiently to guarantee that 
the precious 
items can always fill their slots, but not as 
severely as to lose 
too much value on the paltry items. 
The right volume restrictions in each phase of 
the algorithm are 
communicated via 
a constant amount of advice as well. 
Again, this cannot happen 
directly, since the right 
volume range might be infinitesimally small. 
Instead, the volume is controlled indirectly, 
via bounds not on some size but instead on the value 
that is provided by the paltry 
items packed before filling the current slot 
with a precious item. 

Communicating the necessary volumes in this 
way is possible 
only up to some 
precision $\eps_\round$, and an overestimation could mean the 
loss of a crucial precious item. 
If we always round down, then this cannot happen, but 
the algorithm might reject some paltry item 
it should have 
kept for a selection of 
maximum yield. This is negligible if it happens 
only once, but we 
cannot tolerate 
taking such a loss in every phase, with every new 
volume bound. The 
solution is to analyze the situation using 
a special \emph{splitting slot}, which can accommodate one paltry 
item at the time outside of the knapsack. 
We imagine the splitting slot lending us from 
the item stored in it any desired 
fraction at any time---just always the same fraction of the value 
and size.  
We refer to the item currently stored in the splitting slot as the 
\emph{split item}. 
We will consider an algorithm that maintains a split item of 
highest yield after the remaining paltry items kept in the 
knapsack. 
This imaginary algorithm is an antecedent to our real 
advice algorithm, which builds on it but cannot actually split any 
items of course; we refer to the purely hypothetical precursor as 
the 
\emph{virtual version} of our \emph{actual algorithm}.

The actual algorithm will mimic the virtual version as closely as 
possible and deliver a result that is only marginally worse. 
Whenever the virtual version splits an item, the actual algorithm 
needs to decide whether to discard this item or store it 
completely. 
The challenge is to take the right decisions to allow for all slots 
to be filled in time and also 
avoid an undue 
accumulation of losses by passing on too many 
split items.

The advice helps the actual algorithm by indicating for every 
phase whether an item 
stored in the splitting slot by the virtual version is to be packed 
or discarded. This is done in such a way that 
in the end, the actual algorithm will have relinquished only 
the value of a single paltry item, namely the one kept in the 
splitting slot when the instance ends. 

Packing entire items from the splitting slot comes with problems on 
its own; these paltry items might block for the actual algorithm 
some precious items that are packed by the virtual version. 
This problem is addressed by further advice to the algorithm on how 
to prioritize the packing of precious versus paltry items in each 
phase. 
This advice, telling the algorithm when to \emph{actualize} a 
virtual packing of an item, is based on the solution that the 
virtual version would eventually produce if it existed. 
The virtual version does not depend on the actual algorithm and has 
no need for the part of the advice on actualization, 
which avoids any circular reasoning.  

There are a few more technical 
issues to be dealt with, for example the special 
case that the precious items contribute only marginally to the 
value of the optimal solution. This undermines the estimates for the 
optimal solution value, is thus flagged by a 
dedicated advice bit $b_\purelypaltry$, and handled by switching 
from the elaborate value-based limits that dampen the general 
yield-greedy strategy to a simpler size-sensitive strategy.

\begin{figure}
\begin{tikzpicture}[xscale=.64]
\newcommand{\myheight}{.22}
\newcommand{\mywidth}{*13.8}
\newcommand{\myadjacent}{.35}
\newcommand{\mygap}{.0072\mywidth}
\newcommand{\closedopen}[3]{
\draw[fill=#3] (#2-\mygap,-\myheight) 
-- 
(#1,-\myheight) -- 
(#1,\myheight) -- 
(#2-\mygap,\myheight);
\draw[fill=#3] (#2-\mygap,-\myheight) 
arc 
(-atan(\myheight/\myadjacent):atan(\myheight/\myadjacent):{sqrt(\myheight^2+\myadjacent^2)});
}
\newcommand{\openclosed}[3]{
\draw[fill=#3] (#1+\mygap,-\myheight) 
-- 
(#2,-\myheight) -- 
(#2,\myheight) -- 
(#1+\mygap,\myheight);
\draw[fill=#3] (#1+\mygap,\myheight) 
arc 
(180-atan(\myheight/\myadjacent):180+atan(\myheight/\myadjacent):{sqrt(\myheight^2+\myadjacent^2)});
}

\openclosed{-4}{-3}{white}
\openclosed{-3}{-2}{black!10}
\openclosed{-2}{-1}{white}
\openclosed{-1}{0}{white}
\openclosed{0}{1}{white}
\openclosed{1}{2}{black!10}
\openclosed{2}{3}{white}
\openclosed{3}{4}{white}
\openclosed{4}{5}{white}
\openclosed{5}{6}{black!10}
\openclosed{6}{7}{white}
\openclosed{7}{8}{white}
\openclosed{8}{9}{white}
\openclosed{9}{10}{black!10}
\openclosed{10}{11}{white}
\openclosed{11}{12}{white}
\openclosed{12}{13}{white}
\openclosed{13}{14}{black!10}
\openclosed{14}{15}{white}
\openclosed{15}{16}{white}

\node at (-4.5,0) {$\cdots$};
\node at (-3.5,0) {$V_{-3}$};
\node at (-2.5,0) {$V_{-2}$};
\node at (-1.5,0) {$V_{-1}$};
\node at (-.5,0) {$V_0$};
\node at (0.5,0) {$V_1$};
\node at (1.5,0) {$V_2$};
\node at (2.5,0) {$V_3$};
\node at (3.5,0) {$V_4$};
\node at (4.5,0) {$V_5$};
\node at (5.5,0) {$V_6$};
\node at (6.5,0) {$V_7$};
\node at (7.5,0) {$V_8$};
\node at (8.5,0) {$V_9$};
\node at (9.5,0) {$V_{10}$};
\node at (10.5,0) {$V_{11}$};
\node at (11.5,0) {$V_{12}$};
\node at (12.5,0) {$V_{13}$};
\node at (13.5,0) {$V_{14}$};
\node at (14.5,0) {$V_{15}$};
\node at (15.5,0) {$V_{16}$};
\node at (16.5,0) {$\cdots$};

\draw[decoration={brace}, decorate] 
(-4+\mygap/2,.3) node {} -- (0,.3);
\draw[decoration={brace}, decorate] 
(0+\mygap/2,.3) node {} -- (4,.3);
\draw[decoration={brace}, decorate] 
(4+\mygap/2,.3) node {} -- (8,.3);
\draw[decoration={brace}, decorate] 
(8+\mygap/2,.3) node {} -- (12,.3);
\draw[decoration={brace}, decorate] 
(12+\mygap/2,.3) node {} -- (16,.3);

\node at (-2,.6) {$W_{0}$};
\node at (2,.6) {$W_1$};
\node at (6,.6) {$W_2$};
\node at (10,.6) {$W_3$};
\node at (14,.6) {$W_4$};

\draw[->] (2,-.7) -- (2,-.3);
\node at (2,-1) {$v_1$};
\draw[->] (6.7,-.7) -- (6.7,-.3);
\node at (6.7,-1) {$v_i$};
\draw[->] (11.4,-.7) -- (11.4,-.3);
\node at (11.4,-1) {$v(S)$};

\draw[decoration={brace}, decorate] 
(-5+\mygap/2+\mygap,1.0) node {} -- 
(4-\mygap,1.0);
\node[align=center] at (-0.5,1.8) 
{{}\\paltry and\\provenly paltry};

\draw[decoration={brace}, decorate] 
(4+\mygap/2+\mygap,1.0) node {} -- 
(8-\mygap,1.0);
\node[align=center] at (6,1.8) 
{paltry but\\presumably\\precious};

\draw[decoration={brace}, decorate] 
(8+\mygap/2+\mygap,1.0) node {} -- 
(12-\mygap,1.0);
\node[align=center] at (10,1.8) 
{{}\\{}\\precious};

\end{tikzpicture}
\caption{A schematic illustration 
showing some of 
the infinitely many \emph{classes} 
$V_k$ and 
\emph{comboclasses} $W_k$ used in 
the proof of \cref{thm:general_log_bits_upper}. 
The $x$-axis shows the value 
ranges of the items contained in each 
class; it 
has 
logarithmic scale. We have $0$ 
infinitely far to the left and $\infty$ 
infinitely 
far to the right. Each class 
$V_k$ covers 
values that spread across a factor of 
$1-\varepsilon_\spread$, which 
is the unit 
length in the logarithmic scale. The 
comboclasses are $K=4$ 
units wide, and there are 
$K=4$ modulo classes. The modulo 
class $M_2$ is shaded 
in gray; this is the modulo class 
containing the 
first item, which has value $v_1$. 
The scale is 
shifted such that the following three 
properties 
are satisfied: The value $v_1$ of the first item lies 
right at the top end of its class 
(i.e., the right end in this illustration), 
the class of the first item is part of the 
comboclass $W_1$, 
and the optimal solution value $\val(S)$ lies somewhere  
in the highest class of its comboclass, namely 
in the modulo class $M_0$.
The position of the optimal solution 
value determines what we call 
\emph{precious} and \emph{paltry}. 
The highest value achievable with the 
items seen so far, represented by $v_i$ in the example, 
determines what 
we call 
\emph{presumably precious} and 
\emph{provenly paltry}.}
\label{fig:classes}
\end{figure}

\paragraph*{\textit{Notions and 
Notation}}\label{par:notionsandnotation}
We make the same assumptions as in the proof of 
\cref{thm:prop_near_opt_upper} for the proportional 
problem: The 
knapsack has capacity 
$1$, all items have size at most $1$, and $\eps\le 
1/2$. 
Again, we denote the items in the order of their 
appearance in the 
instance by $1,2,\dots,n$. We write $\size_i=\size(i)$ 
and 
$\val_i=\val(i)$ for the size and value of item 
$i$, respectively. 

To make to proof more understandable, we use the four constants 
$\eps_\klein$, 
$\eps_\paltry$, 
$\eps_\spread$, and 
$\eps_\round$, 
which all depend only on the given parameter $\eps$ but are 
used in different roles. They need to satisfy several 
inequalities; a concrete list of possible choices is 
$\eps_\klein=\eps/2^3$, 
$\eps_\paltry=\eps^2/2^5$, 
$\eps_\spread=\eps^4/2^{14}$, and $\eps_\round=\eps^6/2^{20}$.

The oracle begins by computing an arbitrary 
optimal solution $S$.  
In contrast to the proof for the proportional analogue of our 
theorem, not just the advice but also the 
constructed classes 
depend on $S$ and the instance as well.  
Specifically, the oracle uses the value $\val_1$ of 
the first item 
in the instance and the optimal solution value 
$\val(S)$ to 
partition all items into the bi-infinite 
sequence of value classes
\[
V_k = \Big\{\,i\in\{1,\dots,n\}\ \Big|\ 
\Big\lceil\log_{1-\eps_\spread}\frac{\val_1}{\val_i}\Big\rceil+
 K - 
\Big(\Big\lceil\log_{1-\eps_\spread}\frac{\val_1}{\val(S)}\Big\rceil\bmod{K}\Big)=
k\,\Big\} 
\]
for every integer index $k$, with 
$K=\big\lceil\log_{1-\eps_\spread}\eps_\paltry\big\rceil+1$ being a 
constant for any parameter $\eps$. 
\begin{claim}
These classes are constructed such that they 
simultaneously satisfy the following 
four properties, which are straightforward to 
verify and will 
prove crucial later on.
\begin{enumerate}
\item
Each class 
contains only 
items whose values lie within an interval 
spanning a factor of 
$1-\eps_\spread$ from the upper to the lower interval end. 
\item 
The value $v_1$ of the first item marks 
the upper end 
point of the value interval of the 
class containing this first item. 
\item 
The first item is contained in one of the classes $V_1,\dots,V_K$.
\item 
The index of the class whose interval 
contains the 
optimal solution value $\val(S)$ is a multiple of 
$K$. 
\end{enumerate}
\end{claim}
\begin{proof}
For the first property, it we observe that the item value 
$\val_i$ only occurs once in the entire definition of $V_k$, namely 
within 
the first summand  
$\lceil\log_{1-\eps_\spread}(\val_1/\val_i)\rceil$ as a factor in 
the argument of the logarithm. 
Since 
$\log_{1-\eps_\spread}(\val_1/\val_i)=(\log_{1-\eps_\spread}\val_1)-\log_{1-\eps_\spread}\val_i$,
 this term decreases by $1$ if $\val_i$ is replaced by
$\val_i'=\val_i(1-\eps_\spread)$. Thus the values of 
items in any given class can span factor of up to $1-\eps_\spread$ 
but never more. 

To verify the second property, we observe that 
$\log_{1-\eps_\spread}(\val_1/\val_i)$ is 
zero for $\val_i=\val_1$ but positive whenever $\val_i>\val_1$ and 
thus $\val_1/\val_i<1$. Note that the logarithm base 
$1-\eps_\spread$ is smaller than $1$. 

The third property follows again due to the first summand vanishing 
for $\val_i=\val_1$. This leaves the class index $k$ equal to 
$0+K-j$, where 
$j=\lceil\log_{1-\eps_\spread}(\val_1/\val(S))\rceil$ is an integer 
taken modulo $K$, meaning that we have $j\in\{0,1,\dots,K-1\}$ and 
thus $k\in\{1,\dots,K\}$. 

Finally, for the fourth and last property, it suffices to set 
$\val_i=\val(S)$ 
and see that the minuend becomes identical to the first summand, 
except for being taken modulo $K$, meaning that they combine to 
a multiple of $K$. 
\end{proof}

For every $k\in\{0,\dots,K-1\}$, we define the 
\emph{modulo class} 
$M_k=\bigcup_{j=-\infty}^{\infty}V_{jK-k}$, consisting of every 
$K$th class. 
There are exactly $K$ distinct modulo classes, namely 
$M_0,M_1,\dots,M_{K-1}$; each corresponds to one of 
the possible 
values of the 
modulo term in the definition of $V_k$. 
We could also have defined $M_k$ as 
$\bigcup_{j=-\infty}^{\infty}V_{jK+k}$, but using $V_{jK-k}$ 
instead of $V_{jK+k}$ will turn out to be more convenient 
later on. 

Beside the $K$ modulo classes we define an infinite number of 
\emph{comboclasses}. 
For every integer $k$, comboclass $W_k=\bigcup_{j=0}^{K-1}V_{kK-j}$ 
comprises $K$ consecutive classes. 
This definition is again optimized for notational convenience in 
the remainder of this paper, 
which is also the reason for including into definition of $V_k$ the 
middle summand $K$. 
For example, we can now reformulate the third and fourth property 
more succinctly as follows: 
Item $1$ is contained comboclass $W_1$, and $\val(S)$ 
belongs to modulo class $M_0$.

For any item $i$, we write $V(i)$, $M(i)$, and 
$W(i)$ to indicate 
the class, modulo class, and comboclass in which 
item $i$ lies. 
Thus, $V(i)$ is the integer $k$ such that $i\in 
V_k$, analogously 
$W(i)$ is the integer $k$ satisfying $i\in W_k$, 
and finally $M(i)$ is the integer $k$ satisfying $i\in M_k$.
Note that by their definitions these three terms satisfy the 
general 
relation $V(i)=W(i)\cdot K-M(i)\label{eq:classrelation}$ for any 
item $i$.

An item can be at most as valuable as the entire 
optimal solution.
The comboclass that would contain such a potential item, 
namely the one with index $W(\val(S))$, 
is called \emph{precious}. 
The $K$ 
classes comprising the precious comboclass and 
all the 
items contained in it are called \emph{precious} 
as well. The 
remaining 
items are called \emph{paltry}. 
In the natural way, we denote the sets of 
precious and paltry items 
by 
$V_\precious$ and $V_\paltry$, 
respectively.

We now refer to \cref{fig:classes}, which illustrates the value 
ranges of the classes, comboclasses, and modulo classes and shows 
which are considered paltry and precious, respectively. The notions 
of \emph{provenly paltry} and \emph{presumably precious}, which 
also appear in this figure, will only be 
defined during the discussion of 
the so-called \emph{level system} further down. 

Using \cref{fig:classes}, it is also simple to verify the 
following. 
The highest 
class that could 
contain any paltry item is $V_{(W(\val(S))-1)K}$, and it is 
separated from 
$V_{W(\val(S))K}$, whose value 
range contains $\val(S)$, by $K-1$ classes. 
Since the value range 
of each of these classes spans a factor of 
$1-\eps_\spread$, 
it follows that any paltry item is worth less than 
$(1-\eps_\spread)^{K-1}\val(S)\le\eps_\paltry\val(S)$,
where the inequality holds by the definition of $K$. 
This fact will be used later on.  

We denote by $u_1<\ldots<u_m$ the precious items 
appearing in the 
optimal 
solution 
$S$ in the order as they are 
presented during the given instance. These items 
partition the 
paltry items $V_\paltry$ into $m+1$ subclasses, 
depending on their 
relative appearance time: 
For $j\in\{0,\dots,m\}$, we let 
$V_{\paltry,j}=\{\,i\in V_\paltry\mid u_j<i<u_{j+1}\,\}$
denote the paltry items appearing after $u_j$ 
and before $u_{j+1}$, 
where we use the auxiliary definitions $u_0=0$ 
and $u_{m+1}=n+1$ for two valueless items 
for notational convenience. 

Finally, we partition the items--- independently of the distinction 
between paltry and precious---into the \emph{big} 
ones of size at least $\eps_\klein$, 
denoted 
by 
$V_\gross$, and the \emph{small} ones below this threshold, 
denoted by 
$V_\klein$. 

\paragraph*{\textit{Advice Content}}
In analogy to the proof of 
Theorem~\ref{thm:prop_near_opt_upper} 
for the proportional problem variant, where it is sufficient to 
define a 
finite number of 
instance-independent, 
size-based classes, we would like for the oracle 
to communicate 
through its advice the 
class of each precious item in the optimal 
solution $S$ in 
appearance order. 
This is impossible, however, since the number of 
classes is 
unbounded this time around. 
Instead, we restrict ourselves to conveying the 
information merely 
modulo the constant $K$. We do so by encoding 
modulo classes---of 
which there 
are only finitely many---for all precious items 
into a tuple 
$(M(u_1),\dots,M(u_m))$. 

In addition, the 
advice tells the algorithm into which 
modulo class the first 
item of the instance is falling; that is, it 
encodes $M(1)$. 

Furthermore, for every $j\in\{0,1,\dots,m\}$, the oracle 
encodes 
$\lfloor\val(S\cap V_{\paltry,j})/\val(S)\rfloor_{\eps_\round}$
 for some sufficiently small $\eps_\round>0$; this 
 is the fraction 
 of the 
optimal solution value that is due to paltry items 
appearing in the instance after $u_j$ but before 
$u_{j+1}$, rounded down 
to 
the nearest multiple of $\eps_\round$. We denote 
this fraction by $f_j$. 

There is also one special advice bit $b_\purelypaltry$, telling the 
algorithm whether packing only paltry items yields a 
$(1-\eps_\purelypaltry)$-competitive solution.  

The advice described so far is everything needed for the virtual 
version of our algorithm to work. 
The actual algorithm makes use of the following additional advice, 
which depends on the online output of the virtual version:
For every $j\in\{1,\dots,m\}$, the advice contains a tuple 
$a_j=(a_{j,1},a_{j,2},\dots,a_{j,m_j})$ of numbers from 
$\{0,1,2,\dots,m\}$. Here, the numbers $1,2,\dots,m$ are the 
indices of the $m$ slots 
used by the algorithm for storing precious items. 
The number $0$ is taken as the index of the splitting slot. 
The numbers of the tuple $a_j$ are determined by what the virtual 
algorithm does with the items appearing in phase $j$. 
More precisely, it only matters what the algorithm does with those 
items that will be part of the eventual virtual solution $S'$, 
excluding the one item remaining in the splitting slot in the end. 
If there are no such items during phase $j$, then the tuple $a_j$ 
is empty. 
Otherwise, consider the first such item of phase $j$. 
It will be part of the eventual solution $S'$, meaning that must be 
packed by the virtual algorithm upon appearance. 
The first number of the tuple $a_j$ indicates the slot into which 
this item is stored, namely slot $a_{j,1}$. 
The next number of the tuple is analogously determined by the next 
item that is packed during phase $j$ and a full part of the 
eventual virtual solution, and so on.  
In summary, tuple $a_j$ contains the indices of the slots filled during phase $j$ with items that will be a complete 
part of the eventual virtual solution, and the indices are sorted in 
the order these slots were filled with said items. 

\paragraph*{\textit{Constant Advice}}
We prove that a constant amount of advice 
suffices to encode the 
information just described. 

We assume without loss of generality that the following, 
separately discussed pieces of information are combined 
into the advice string in such a way that they can be 
unambiguously retrieved. 
This is achieved by means of a suitable self-delimiting encoding 
that increases the advice length by at most a constant factor. 

The number $m$ of precious items in the optimal 
solution, and thus also in the virtual solution, is 
bounded from above by the constant $1/\eps_\paltry$ since 
every precious 
item is worth at least 
$\eps_\paltry\cdot \val(S)$ and the total value of packed precious 
items cannot 
exceed 
$\val(S)$. 

There are only constantly many modulo 
classes, namely 
$K=\big\lceil\log_{1-\eps_\spread}\eps_\paltry\big\rceil+1$; 
a constant amount of advice therefore suffices  
to communicate the modulo classes of the 
$m$ precious items in the optimal solution and 
that of the first item 
of the instance. 

For every $j\in\{0,1,\dots,m\}$, we can encode $f_j$ 
via 
$\lfloor(1/\eps_\round)\val(S\cap\bigcup_{k=0}^jV_{\paltry,k})/\val(S)\rfloor$;
 this is the fraction to 
be approximated divided by $\eps_\round$ and 
rounded down to the 
nearest integer. These integers are clearly bounded by 
the constant 
$1/\eps_\round$. 
Since both $1/\eps_\round$ and 
$m$ are constants, this information is also 
encodable within a constant amount of advice. 

Finally, there are the $m+1$ tuples $a_0,a_1,\dots,a_m$, each 
containing at most $m+1$ numbers ranging from $0$ to $m$. 
A number of advice bits cubic in $m$, but still constant for any 
given parameter $\eps$, is 
enough to store this information too. 

\paragraph*{\textit{Algorithm Description}}
\subparagraph*{A Single Special Switch.}
As a first simple special case, we consider what the algorithm does 
when the bit $b_\purelypaltry$ is set to $1$, meaning that packing 
only paltry items in a yield-greedy fashion would lead to a 
$1/(1-\eps_\purelypaltry)$-competitive solution. 
In this case, the algorithm does not directly implement this 
strategy---being unsure what items are in fact paltry, it is in 
fact unable to do so---but instead runs a yield-greedy strategy on 
the small items, discarding any big item without consideration. 
When analyzing the algorithm later on, we 
will show that this still secures a solution of sufficient value. 

For the remainder of this description, we assume that 
the 
switch $b_\purelypaltry$ is set to zero. 

\subparagraph*{Computing the Classes.}
Immediately after seeing the first item, the 
algorithm uses the 
advice-given information about the modulo class 
of this item and the fact that 
its value $v_1$ lies right at 
the upper end 
of the value interval of its class to 
reconstruct the precise value 
intervals of every class $V_k$. 
This is possible thanks to the four properties of these value 
classes already proved above:

Once the exact value range covered by the class containing the 
first item is known to the algorithm, the first property can be 
used to achieve the same for all other classes.
The remaining three properties help the algorithm retrieve all 
information about the value interval of the class of the first 
item. 
The second property fixes the upper end of this interval to $v_1$. 
The third property tells us that $v_1$ is contained in one of the 
classes $V_1,V_2,\dots,V_K$ of comboclass $W_1$. 
The fourth property determines which of these $K$ classes it is by 
fixing the modulo class of $\val(S)$ and thus also the one of 
$\val_1$, which is included in the advice given to the algorithm. 

As a consequence, the algorithm is able to 
compute $V(i)$, $M(i)$, and $W(i)$ for any 
item $i$ of the instance as soon as it appears. 
We emphasize that, nevertheless, it is still unable 
to tell for 
sure whether an item is precious 
or paltry as long as the instance has not ended. 
Instead, the algorithm partitions the items seen 
so far into 
\emph{presumably precious} and \emph{provenly 
paltry} ones. 
This partition may change over time and is based 
on the following 
\emph{level system}. 

\subparagraph*{The Level System.} 
At any time, there is exactly 
one comboclass 
considered presumably 
precious; all lower comboclasses are provenly 
paltry. 
The designation as presumably precious or 
provenly paltry is naturally inherited by the 
classes and items 
contained in a comboclass. 
For every integer $\ell$, we say that the 
algorithm is at 
\emph{level} $\ell$ if 
$W_\ell$ is the presumably precious comboclass at 
this moment. 
The algorithm starts at level $1$, which means that the 
comboclass containing the 
first 
item of the instance is considered presumably 
precious in the beginning. 
While processing its input, the algorithm may level 
up many times, 
but 
the level is never 
decreasing: 
What has been recognized as presumably precious 
may become provenly paltry at a later 
point, but never the other way around.
There are exactly two events that can trigger 
the raise to a new 
level. 

\begin{enumerate}
\item The first one is the appearance of an item with 
a value above the 
value range of the comboclass that is currently 
seen as presumably 
precious. 
In this case, the comboclass containing this new 
item becomes the 
presumably precious one. 

\item 
The algorithm might also be triggered into leveling up by the 
arrival of a 
provenly paltry item, but only under certain conditions. 
Namely, the algorithm is always keeping track of what 
the optimal solution 
value achievable with the items presented so far, 
including the rejected 
ones, would have been. 
This only hypothetically realizable value is 
continually compared 
against some upper estimate 
$U_\ell$ for 
the optimal solution value $\val(S)$; the exact 
estimate is 
explained and examined in 
the next subsection. 
As soon as the estimate $U_\ell$ for $\val(S)$ is 
exceeded by the 
hypothetically achievable solution value, the 
algorithm enters the 
next level: 
The presumably precious comboclass becomes 
provenly paltry and the 
comboclass immediately above it becomes the new 
presumably 
precious one. 
\end{enumerate}

We will explain the precise purpose for employing exactly these two 
triggers when analyzing the competitive ratio of the algorithm. 
Essentially, the first one protects precious items from being 
discarded due to procrastinated level changes, while the second 
trigger's 
paramount duty is to prevent too many paltry items from being 
misjudged as precious, which can also lead to them being unduly 
discarded. 

Each of the two triggers alone would already ensure that the 
algorithm eventually reaches a level such 
that the presumably 
precious comboclass is in fact the precious one. 
We call the level for which this happens the 
\emph{final level} and refer to all previous ones as \emph{the lost 
levels} for a reason that will become clear later on.

\subparagraph*{Estimating the Optimal Solution Value.} 
The upper estimate for $\val(S)$ mentioned before is 
\[U_\ell=v_1\sum_{j=1}^m(1-\eps_\spread)^{V(1)+M(u_j)-\ell 
K}/(1-\eps_\round(1+1/\eps_\paltry)-\sum_{p=0}^mf_p),\]
 where $\ell$ is the index of the comboclass 
 currently considered 
 presumably 
precious by the algorithm. The modulo class index 
$M(u_j)$, the 
class index $V(1)$, 
and the rounded fractions 
$f_p=\lfloor\val(S\cap 
V_{\paltry,p})/\val(S)\rfloor_{\eps_\round}$
of the 
optimal solution 
value contributed by paltry items from phase $p$ are all 
retrieved from the 
advice. 
When calculating the estimate above, the algorithm 
assumes to have 
reached its 
final level; that is, 
it operates under the assumption that the 
presumably precious 
comboclass $W_\ell$ is in 
fact the precious one. 
Under this assumption, $U_\ell$ is indeed a valid 
upper estimate 
for $\val(S)$ as we will show now. 

We have $\val(S)=\val(S\cap V_\precious)+\val(S\cap 
V_\paltry)$ 
and begin by analyzing the first of the two summands. 
It is the value 
contributed to the optimal solution by the $m$ 
precious items contained in 
it. 
The algorithm can use the information about the 
modulo classes of 
these $m$ items to place the value of each of 
them into an interval 
spanning a factor of $1-\eps_\spread$.  
Specifically, we start with the already established general 
relation 
$V(i)=W(i)K-M(i)$, 
apply it to $i=u_j$, 
and use the assumption $W(u_j)=\ell$ of the 
current comboclass 
being the precious one to 
obtain $V(u_j)=\ell K-M(u_j)$. 
Since $v_1$ lies at the upper endpoint of 
class $V(1)$, we 
can compute the upper end of the value range of 
class $V(u_j)$ as 
$v_1/(1-\eps_\spread)^{V(u_j)-V(1)}=v_1(1-\eps_\spread)^{V(1)+M(u_j)-\ell
K}$ precisely. Note that the given advice 
and the value of 
the first 
item are in tandem indeed sufficient to 
calculate this number. It 
is an upper bound on 
$\val(u_j)$ such that the true value cannot be smaller than 
a factor $1-\eps_\spread$ of it, provided that the algorithm's 
assumption 
$W(u_j)=\ell$ holds 
true. If the 
algorithm is mistaken, however---that is, if $\ell 
<W(u_j)$---then 
the estimate is far too 
low; namely, it would be exactly a factor 
$(1-\eps_\spread)^{(W(u_j)-\ell)K}\le(1-\eps_\spread)^K$ of what 
it otherwise would have been. 
If the algorithm has not yet reached the final level, 
the estimate for 
$\val(u_j)$ is 
therefore at most a fraction 
$(1-\eps_\spread)^{K-1}\le \eps_\paltry$ of the true 
value. As mentioned before, 
the inequality holds by the definition of $K$. 
All of the assertions above for any single 
precious item naturally 
carry over to the sum
\[t=\sum_{j=1}^mv_1(1-\eps_\spread)^{V(1)+M(u_j)-\ell 
K}.\] 
In particular, the total value $\val(S\cap 
V_\precious)$ of all precious 
items in the 
fixed optimal solution $S$---the first of the two summands to be 
examined---lies in the interval 
$[(1-\eps_\spread)t,t]$ during the 
final level. 

Using the advice information 
$f_p=\lfloor\val(S\cap 
V_{\paltry,p})/\val(S)\rfloor_{\eps_\round}$ on the value of the 
paltry items in the optimal solution appearing during each phase 
$p\in\{0,1,\dots,m\}$---given as fractions of the optimal solution 
value, rounded down 
to the nearest multiple of $\eps_\round$---the 
algorithm  
can 
narrow 
down the second summand $\val(S\cap V_\paltry)$ to 
$g\cdot \val(S)$ 
for some factor $g\in 
[\sum_{p=0}^mf_p,\sum_{p=0}^m(f_p+\eps_\round)]\subseteq 
[\sum_{p=0}^mf_p,\eps_\round(1+1/\eps_\paltry)+\sum_{p=0}^mf_p]$,
 where 
we have used once more the bound $m\le1/\eps_\paltry$, derived via 
the total value of all precious items. 
Together with the bounds found for 
the first summand in the previous paragraph we obtain the 
following two inequalities, in which the optimal solution value 
$\val(S)$ occurs three times: 
\[(1-\eps_\spread)t+\val(S)\sum_{p=0}^mf_p\le\val(S)\le 
t+\val(S)\big(\eps_\round(1+1/\eps_\paltry)+\sum_{p=0}^mf_p\big).\]

From this we obtain for $\val(S)$ a lower bound 
of $L_\ell=(1-\eps_\spread)t/(1-\sum_{p=0}^mf_p)$ and the presumed 
upper bound 
$U_\ell=t/(1-\eps_\round(1+1/\eps_\paltry)-\sum_{p=0}^mf_p)$, 
which is 
an actual 
upper 
bound only during the last level. 
The upper and lower bound lie a factor 
\begin{align*}
\frac{U_\ell}{L_\ell}={}&\frac{1-\sum_{p=0}^mf_p}{(1-\eps_\spread)(1-\eps_\round(1+1/\eps_\paltry)-\sum_{p=0}^mf_p)}\\
\le{}&\frac{1-\sum_{p=0}^mf_p}{1-\eps_\spread-\eps_\round(1+1/\eps_\paltry)-\sum_{p=0}^mf_p}\\
={}&1+\frac{\eps_\spread+\eps_\round(1+1/\eps_\paltry)}{1-\eps_\spread-\eps_\round(1+1/\eps_\paltry)-\sum_{p=0}^mf_p}
\end{align*} apart. 
Note that this quotient no longer depends on 
$\ell$.
Furthermore,we may assume that 
$\sum_{p=0}^mf_p<1-\eps_\purelypaltry+\eps_\paltry$. 
This is because 
otherwise the algorithm 
could attain a fraction 
$\sum_{p=0}^mf_p-\eps_\paltry\ge 1-\eps_\purelypaltry$ of the 
optimal 
solution 
value by packing paltry items exclusively in an otherwise purely 
yield-greedy 
manner uninhibited by any value limit; 
this special 
case is covered by the dedicated advice bit $b_\purelypaltry$, 
which, when set to 
$1$, switches the restrictions applied to the algorithm's 
yield-greedy 
approach from a value limit to a size limit, as already explained 
at the beginning of this proof. 
We can therefore use 
$1-\sum_{p=0}^mf_p>\eps_\purelypaltry-\eps_\paltry$ 
to 
obtain 
\begin{align*}
\frac{U_\ell}{L_\ell}
\le{}&1+\frac{\eps_\spread+\eps_\round(1+1/\eps_\paltry)}{\eps_\purelypaltry-\eps_\paltry-\eps_\spread-\eps_\round(1+1/\eps_\paltry)}\\
={}&\frac{\eps_\purelypaltry-\eps_\paltry}{\eps_\purelypaltry-\eps_\paltry-\eps_\spread-\eps_\round(1+1/\eps_\paltry)}\\
\intertext{and thus}
\frac{L_\ell}{U_\ell}\ge{}&\frac{\eps_\purelypaltry-\eps_\paltry-\eps_\spread-\eps_\round(1+1/\eps_\paltry)}{\eps_\purelypaltry-\eps_\paltry}\\
={}&1-\frac{\eps_\spread+\eps_\round(1+1/\eps_\paltry)}{\eps_\purelypaltry-\eps_\paltry}\\
\ge{}&1-2\eps_\spread/\eps_\purelypaltry,
\end{align*}
where, for the 
last inequality, we used $\eps_\spread\ge 
\eps_\round(1+1/\eps_\paltry)$,
which is satisfied by our choices of 
$\eps_\spread=\eps^4/\eps^{14}$, $\eps_\round=\eps^6/2^{20}$, 
and $\eps_\paltry=\eps^2/2^5$.  
In other words: During the final level the algorithm's 
estimate 
$U_\ell$ is indeed an 
upper 
bound on $\val(S)$ exceeding the real value by 
at most 
$U_\ell\cdot 2\eps_\spread/\eps_\purelypaltry$.
The algorithm can thus derive from any hypothesized upper bound 
$U_\ell$ 
for $\val(S)$ an actual lower bound 
$U_\ell(1-2\eps_\spread/\eps_\purelypaltry)$.
We will put this insight to use when analyzing 
the competitive ratio of the algorithm. 

\subparagraph*{The Slot System.}
In the precious part of the knapsack, the 
algorithm uses a 
\emph{slot 
system} to store 
presumably precious items. 
There are $m$ \emph{slots}, each of which can accommodate 
one such item, and all of them are \emph{empty} in the 
beginning. They will be 
\emph{filled} one by one, each with exactly one 
item. The items in 
filled slots may also 
be exchanged at some point. 
In contrast to what happens in the simpler 
algorithm for 
the proportional version of the problem, slots might also 
be emptied again when the algorithm decides to level up. 
The second difference is that there are now two stages to filling a 
slot.
Any slot first has to be filled \emph{virtually} before this 
filling is \emph{actualized} at some point. The precise meaning of 
these two notions is explained later on. 
The order in which the slots are filled virtually is determined by 
the modulo classes of the precious items $u_1,\dots,u_m$ in the  
optimal solution: 
Slot $j$ must be filled virtually 
after slots 
$1$ through $j-1$ and before the other 
slots. As soon as this condition is met, slot $j$ is filled 
virtually with the first \emph{matching item}, that is, the first 
item from the 
same modulo class as $u_j$.
The order in which $m$ slots receive an actualized filling is 
encoded in the tuples $a_0,a_1,\dots,a_m$ given 
by the advice. 
A filling may be actualized long after or right when it has 
happened virtually, just never before that. 

When exactly a slot is filled virtually, instead of just in which 
order, is described in the following paragraph; when exactly the 
fillings are actualized is explained after the discussion on 
packing the paltry part of the knapsack.

\subparagraph*{The Phase Progression.}
We say that the algorithm is in \emph{phase} $p$ 
if slot $p$ has 
been filled virtually but not yet slot $p+1$. 
Since all precious items are discarded when the 
algorithm is 
leveling up, all slots are empty again at the 
start of a new level, 
implying that the algorithm's phase is reset to 
$0$ as well. 
In general, the current phase $p$ ends as soon as the 
instance ends or presents a 
presumably precious item $i$ 
satisfying the following two conditions. 
On the one hand, it is from the same modulo class 
as 
$u_p$---recall that we can compute $M(i)$ for 
any given item and 
know $M(u_p)$ from the 
advice. On the other hand, the virtual algorithm can 
fit it into the 
precious part of the 
knapsack beside the items 
already present in the other slots, even if just virtually, and the 
packed paltry items, 
including the currently used fraction of any split item. 
If 
both of these conditions are 
met, said item $i$ is 
virtually packed into slot $p$. For any other presumably 
precious item $i$, 
it is first 
checked whether there are \emph{matching slots}---that is, slots 
reserved for items from modulo class $M(i)$---filled 
with items larger 
than 
$i$. If this is the case, a largest of these 
items is evicted, either just virtually or actually, and 
item $i$ takes its 
place. Otherwise, the presumably precious item $i$ 
is discarded.

\subparagraph*{Packing the Paltry Part.}
The paltry items in the instance are in 
principle packed in a 
yield-greedy manner, that is, optimizing the ratio 
of value to size 
on the paltry items. 
If there were no precious items, this would 
already guarantee a 
sufficient approximation to the optimal solution 
value; we would 
lose at most the value of a single paltry item. 
If there are precious items, however, then we need 
to restrict the 
reservation of paltry items such that they never 
block a crucial 
precious item from filling an empty slot. But too 
much restriction 
is bad as well; if we always favor the 
packing of precious 
items matching an empty slot over accruing 
paltry items, then we 
might discard too many paltry items; after all, it 
is impossible to 
give 
exact bounds on the sizes of the optimal 
precious items, and the 
paltry items can be infinitesimally small. 

The right type of restriction keeps the total 
size of the reserved 
paltry items low enough to allow for filling the 
next empty slot 
when the right precious items appears but not 
more than a constant 
fraction below the optimal amount. 
Directly encoding the right volume bound into 
the advice is 
impossible, however, since the requisite precision grows 
indefinitely with decreasing item sizes.
Instead, we can limit the value that we should be 
accumulating in the form of 
paltry items in each phase. This does not 
immediately bound the 
size occupied by these items---we have no 
knowledge about the 
necessary yield---but will always do so just in time to 
fill all regular slots with precious items. 

The value goal for each phase is indicated with 
a precision of 
$\eps_\round$ and always rounded down. Additionally, this value 
goal is indicated relative to $\val(S)$, for which the algorithm 
has a lower bound $L_\ell$ that is too small by at most a factor of 
$2\eps_\spread/\eps_\purelypaltry$ as seen above. 
Since there are at most $1/\eps_\paltry+1$ phases, 
the 
value lost due 
to this rounding can be bounded by  
\begin{align*}\label{eq:paltry}
(1/\eps_\paltry+1)(\eps_\round+2\eps_\spread/\eps_\purelypaltry)\val(S)\le
(2\eps_\round+4\eps_\spread/\eps_\purelypaltry)\val(S)/\eps_\paltry,
\end{align*}
which we can still tolerate, as shown by the competitive 
analysis in the end of this proof.  
However, the situation could be exacerbated by the 
fact that when 
discarding a single paltry item to get below the 
rounded value 
threshold, we could incur a loss of up to 
$\eps_\paltry\val(S)$, the 
maximum value of a paltry item, instead of only 
$(\eps_\round+1+1/\eps_\paltry)\val(S)$ in every phase. There are 
up to 
$1/\eps_\paltry+1$ phases, so this could sum up 
to a total loss just shy of $\val(S)$. 
In the following, we introduce what might be called a virtual 
version of our algorithm as a 
tool to circumvent this problem. 

\subparagraph*{The Virtual Algorithm.}
Beside the actual advice algorithm producing our online solution, 
we consider a 
\emph{virtual version} of the algorithm that has one huge, if 
imaginary, advantage: It may use a so-called \emph{splitting slot}. 
The splitting slot can store one paltry item at any point, and the 
virtual algorithm may then split off from the item currently stored 
there any 
desired fraction to be used in the packing of the knapsack. 
The fraction used of the \emph{split item} is allowed to change at 
any moment.  
However, the algorithm must always take the same fraction for the 
value and the size of an item, meaning that 
the split-off part retains the yield of the full item. 

In other words, the algorithm can, for an 
arbitrary positive $r\le1$, treat 
an item $i$ of size $\size_i$ and value $\val_i$ in the splitting 
slot 
as though it were an item of size 
$r\cdot\size_i$ and value, 
$r\cdot\val_i$
and the 
$r$ can be adjusted arbitrarily at any time. 

This enables the virtual algorithm to hit the previously described 
target value 
precisely and persistently, once it could have exceeded it, by 
maintaining in the 
splitting slot a paltry item 
of highest yield after those actually kept outside the splitting 
slot, and always taking just the 
right fraction to stay at the value limit as long as possible. 
The unused fraction is assumed to neither take up any space nor contribute any value. 

There are two ways for an item to be forced out of the splitting 
slot.  
If items of higher yield---or identical yield but smaller 
ones---and sufficient total value appear, the 
present split item will become useless to the virtual algorithm, 
which then discards it for good. 
If, however, there is a phase change accompanied by an increase of 
the value limit, 
then the algorithm may decide to retrieve the entire item from the 
splitting slot, pack it in 
the knapsack regularly, 
and store an item of lower yield---or identical yield 
but lower size---in the splitting slot instead. 

Having described the behavior of the virtual version of the 
algorithm, which uses the imagined capabilities of the splitting 
slot, we now explain how the actual algorithm acts based on the 
decisions taken by the virtual version. 

\subparagraph*{The Actual Algorithm.}
The actual algorithm is essentially trying to trace every  
action of 
its virtual version. 
However, it is not allowed to split any items and thus faced 
with a binary decision 
whenever the virtual algorithm stores an item in the splitting 
slot: 
either pack said item as a whole 
or discard it completely. 
Both choices come with their own set of problems. 

It is the simpler option for the algorithm to discard any split 
item, as it then can still imitate the virtual version in every 
other respect, by operating as though any fraction from the split 
item were available at any time. In particular, the precious items 
would be handled absolutely identically by both versions of the 
algorithm. However, if too many split items kept by the virtual 
algorithm are discarded, then their total value could accumulate to 
a substantial part the optimal solution value. 
Thus it will be necessary to store split items under certain 
circumstances. 

This opens up the possibility for the second type of error, 
entailing even more serious 
consequences. If the actual algorithm packs an item 
in its entirety instead of only the fraction used 
by the virtual version, then there is less space left 
for precious items needed to fill up all slots in the end. 
The naïve approach of just following the virtual version in 
packing a precious item whenever feasible---that is, unless packed 
paltry items, including the used fraction of the split item, would 
need to be discarded---does not work here. 
This could entice the algorithm into filling a slot with a precious 
item much larger than what this slot will contain in the end, thus 
preventing the packing of rather large but still important paltry 
items into the 
splitting slot. 

Instead, the algorithm uses the advice tuples $a_0,a_1,\dots,a_m$ 
to navigate the two issues in such a way that its final output is 
in fact identical to that of the virtual version, except for 
missing the one item remaining in the splitting slot at the end.

Specifically, the advice tuple 
$a_j=(a_{j,1},a_{j,2},\dots,a_{j,m_j})$ is used in phase $j$ as 
follows. 
The algorithm does not pack any precious items or split items by 
default. 
Once the virtual algorithm packs an item into slot $a_{j,1}$, 
however, the actual algorithm tries to follow suit. 
It might be unable to do so if it has stored items of a larger total 
size due to its inability to split items. 
If it can fill a new item into slot $a_{j,1}$ alongside with 
the virtual version, however, then from this moment on, the 
algorithm starts actually packing into slot $a_{j,1}$ whatever the 
virtual version packs into it; we say that slot $a_{j,1}$ is now 
\emph{actualized}. 
After actualizing slot $a_{j,1}$, the algorithm starts monitoring 
slot $a_{j,2}$, which might already be filled virtually at this 
point, but is not yet filled actually. Whenever a new item---assume 
that it is precious for the moment---is 
stored into $a_{j,2}$ by the virtual algorithm, the actual 
algorithm observes if it could do the same. As soon as such an 
opportunity arises, it is taken, and slot $a_{j,2}$ is now also 
actualized. This continues until the tuple is exhausted, which 
means that no more slots will be actualized in this phase. 
If the advice tells the algorithm to actualize the splitting slot, 
which has the number $0$, with a paltry item, this all works 
identically, with two exceptions:
First, if the virtual algorithm updates an already actualized item 
in the splitting slot, 
then all slots that were actualized after the splitting slot in the 
current phase are 
virtualized again, meaning that the actual algorithm discards the 
items stored in them. 
And second, the splitting slot is reverted to be virtual at the 
beginning of each 
new phase. Nevertheless, an item already stored in the splitting 
slot at the beginning of a phase remains there until it is replaced 
by a future item. 

We have now fully described our algorithm, whose pseudo-code is 
found in \cref{alg:general_log_bits_upper}. 
Note that both the value-restricted and the size-based 
yield-greedy procedure used for 
packing the paltry items are implemented by the function 
\textsc{Greedy}. 

Turning now to the analysis of the algorithm, 
we will first show that the algorithm performs sufficiently in the 
special mode where items are filtered based on their size.
Then we prove that the virtual version eventually fills all slots 
and always attains the value goal on the paltry items, resulting in 
a total value that surpasses our final goal with more than 
$\eps_\paltry\val(S)$ to spare. 
Finally, we prove that the actual 
algorithm can successfully re-create the solution by the virtual 
version up to potentially one single paltry item, which we can 
forego without problem due to the leeway afforded by the virtual 
version.

\paragraph*{\textit{Success of Size-Based Strategy.}}
We first consider the rather simple case that the special bit 
$b_\purelypaltry$ is set to $1$. The oracle does this if and only 
if a purely yield-greedy approach on the paltry items, ignoring all 
precious ones, is guaranteed to achieve a competitive ratio of 
$1/(1-\eps_\purelypaltry)$ on the given instance. 
However, as mentioned before, the algorithm does not know for sure 
which items will turn out to be paltry in the end; thus it is 
implementing a size-based strategy instead in this case, packing 
only small items and dismissing the big ones.

Since the big items all have size at least $\eps_\klein$, at most 
$1/\eps_\klein$ of them fit into any feasible solution, in 
particular the $1/(1-\eps_\purelypaltry)$-competitive solution that 
would 
have been produced by greedily packing only paltry items. 
Omitting from this solution, which contains only paltry items, the 
up to $1/\eps_\klein$ big items decreases its value by less than 
$\eps_\paltry/\eps_\klein$. This means that we now have a solution 
containing only small items, worth at least a fraction 
$1-\eps_\purelypaltry-\eps_\paltry/\eps_\klein$ of the optimal 
solution 
value. Running a yield-greedy strategy on the small items either 
reproduces this solution perfectly or it leaves, as soon as a 
packed item is discarded, an unfilled gap smaller than 
$\eps_\klein$. 
Since the yield is optimized for the filled part of the knapsack, 
this is a solution worth at least a fraction 
$(1-\eps_\purelypaltry-\eps_\paltry/\eps_\klein)(1-\eps_\klein)\ge 
1-\eps_\purelypaltry-\eps_\paltry/\eps_\klein-\eps_\klein$ of the 
optimal 
solution value, which for our choices of $\eps_\klein=\eps/2^3$
and 
$\eps_\paltry=\eps^2/2^5$ is greater than $1-\eps/2\ge 1/(1-\eps)$. 
We 
can therefore assume $b_\purelypaltry=0$ from now on.

\paragraph*{\textit{Success of Virtual Version.}}
Assuming that the algorithm already starts in its final 
level---which is 
proved to be possible without loss of generality in its own section 
on leveling up---we show that the virtual version completes all 
phases and hits its value goal on the paltry items precisely in 
each phase. 
Recall that after reaching the final level, the notions of 
\emph{paltry} and 
\emph{precious} 
coincide with \emph{provenly paltry} and \emph{presumably 
precious}, 
respectively. 

Denote by $\val_{\paltry,j}=\val(S\cap V_{\paltry,j})$ and $\size_{\paltry,j}=\size(S\cap 
V_{\paltry,j})$ the total value and size, respectively, of the items in optimal solution $S$ 
appearing after $u_j$ but before $u_{j+1}$.
Recall that the algorithm computes $f_j\cdot L_\ell$---where $f_j=\lfloor\val(S\cap 
V_{\paltry,j})\rfloor_{\eps_\round}$ is given directly by the advice and $L_\ell$ is the 
current lower bound on $\val(S)$---as a fairly accurate lower bound on $\val_{\paltry,j}$. 
Denote this lower bound by $\val'_{\paltry,j}=f_j\cdot L_\ell$, and denote by 
$\size'_{\paltry,j}$ the proportional lower bound on $\size_{\paltry,j}$, namely 
$\size'_{\paltry,j}=\size_{\paltry,j}\cdot 
\val'_{\paltry,j}/\val_{\paltry,j}$. 
The desired statement follows from a two-part claim by 
induction 
over 
$p\in\{0,1,\dots,m+1\}$, where we use once more our auxiliary 
definition of two imaginary items $u_0=0$ and 
$u_{m+1}=n+1$.

\subparagraph*{Induction Claim.} 
\begin{enumerate}
\item 
After the decision on item $u_p$, the first $p$ slots are filled, 
whether just virtually or 
already actualized, 
with items of size at most 
$\sum_{j=1}^p\size(u_j)$. 
\item 
Consider the paltry items in the knapsack when $u_p$ is presented, 
including the used 
fraction of any split item. Among these, select items in a manner  
greedy for high yield and secondarily for low 
size---allowing for any items to be split, regardless 
of whether they are in the splitting slot---and stop only if 
the total value reaches $\sum_{j=0}^{p-1}\val'_{\paltry,j}$. 
The value of this selection is in fact exactly 
$\sum_{j=0}^{p-1}\val'_{\paltry,j}$ and its size at most 
$\sum_{j=0}^{p-1}\size'_{\paltry,j}$. 
\end{enumerate}

\subparagraph*{Base Case.} 
The base case $j=0$ is trivially true:
Phase $0$ begins 
with the decision on the auxiliary item $u_0=0$, which does not fit 
any of the slots, which thus all remain empty. 
The second part claim is empty as well due to 
$\sum_{j=0}^{-1}\val'_{\paltry,j}=\sum_{j=0}^{-1}\size'_{\paltry,j}=0$.
 
\subparagraph*{Induction Step.}
Let an arbitrary $p\in\{0,1,\dots,m\}$ be given. 
We assume the claim for all $p'\in\{0,1,\dots,p\}$ as our 
induction hypothesis and use it to prove the statement for $p+1$. 
By the first part of the hypothesis, the algorithm is at least in 
phase $p$ when the items between $u_p$ and $u_{p+1}$ are presented. 
Therefore, the value limit for packing paltry items is already 
$\sum_{j=0}^{p}\val'_{\paltry,j}$ or higher at the beginning of 
this period of the instance. 
The set $V_{\paltry,p}$ of paltry items presented during this 
period contains a set of items, namely $S\cap V_{\paltry,p}$, with 
total size $\val_{\paltry,p}$ and size $\size_{\paltry,p}$. 
A yield-greedy online packing choosing from these items with a 
value limit 
$\val'_{\paltry,p}$ achieves, if splitting items arbitrarily is 
allowed, a selection with a value of exactly $\val'_{\paltry,p}$ 
and---since the selection has the highest possible yield---size at 
most $\size'_{\paltry,p}$. 
We know from the second part of the hypothesis that at the 
beginning of the considered period the knapsack has already 
contained paltry items allowing for a selection of value 
$\sum_{j=0}^{p-1}\val_{\paltry,j}$ and size at most 
$\sum_{j=0}^{p-1}\size_{\paltry,j}$. 
The combination of this selection with the previously described one 
has value $\sum_{j=0}^p\val'_{\paltry,j}$ and some size which is 
at most 
$\sum_{j=0}^p\size'_{\paltry,j}$. Our algorithm will therefore also 
achieve at the end of the considered period a selection of this same
value and some size that might be even lower---since it has to 
chance to use some items in $V_{\paltry,p}\setminus S$ of 
potentially greater yield---but never 
higher. 
This concludes the proof of the second part of the claim.

For the first part, it suffices to consider the case that the first 
$p+1$ slots do not yet contain items of size at most 
$\sum_{j=1}^{p+1}\size(u_j)$ when item $u_{p+1}$ is presented. 
By the induction hypothesis, the first $p$ slots already contained 
items of 
size at most $\sum_{j=1}^{p}\size(u_j)$ after the decision on 
$u_j$. 
This remains true until the instance ends since items in a 
slot are only ever replaced by smaller ones. 
If slot $p+1$ is already filled when $u_{p+1}$ is presented, then 
the item in this slot either has size at most $\size(u_{p+1})$, or 
it is replaced by $u_{p+1}$, both of which mean that the induction 
claim for $p+1$ is satisfied. 
Finally, if slot $p+1\le m$ is still empty upon the presentation of 
$u_{p+1}$, then we can combine the first part of our induction 
hypothesis and the already proven second part for $p+1$ to bound 
the total size of all packed items at this moment by 
\[\sum_{j=0}^p(\size(u_j)+\size'_{\paltry,j})\le 
\sum_{j=0}^p(\size(u_j)+\size_{\paltry,j})\le 
\sum_{j=0}^{m}(\size(u_j)+\size_{\paltry,j})-\size(u_{p+1})
=\size(S)-\size(u_{p+1}).\] 
Since $S$ is a feasible solution, item $u_{p+1}$ will fill slot 
$p+1$, which concludes the proof of the second part of the claim,  
the induction step from $p$ to $p+1$, and thus the entire induction.

\paragraph*{\textit{Success of Actual Algorithm}}
Having seen that the virtual version of the algorithm performs 
sufficiently well---the exact competitive analysis follows later 
on---we now prove that the actual algorithm successfully mimics 
this behavior in the following sense: 
Despite deviating from the virtual online solution by 
non-negligible values at some times, the actual algorithm produces 
a final output that is identical to that of the virtual solution 
$S'$, except for the omission 
of one single paltry item, namely the one remaining in the 
splitting slot when the instance ends. 
We will see that the actual algorithm, which uses exactly the 
same phases and value 
limits as the virtual version, can use the advice to actualize  any 
item that is part of the eventual solution immediately when is 
appears. 
The first key toward this is the following observation. 
During any given phase, the value limit used by the algorithm 
remains fixed. 
Hence, any item that is split at any point during a given 
phase either remains in the splitting slot until the end of this 
phase or---if it leaves the splitting slot before the phase has 
ended---it is discarded in favor of 
items of greater yield or identical yield but larger size, as 
dictated by the yield-greedy approach for paltry items. 
This means that the paltry items remaining in the knapsack at the 
end of a phase, excluding the one in the 
splitting slot at this point, have never been in the splitting slot 
during the entire phase.
By not actualizing the splitting slot, the algorithm will therefore not miss during this 
phase any paltry item from the eventual virtual solution $S'$ but potentially the one item 
that is in the splitting slot at the end of the phase. Moreover, 
actualizing the splitting slot only 
upon appearance of this one item suffices to not miss any paltry item from $S'$ in this 
phase.
The algorithm can therefore focus on actualizing the precious items and paltry items that 
are in the splitting slot at the end of some phase---although it
does not know beforehand which of the paltry items will turn out to 
be of this kind.  

The second key, addressing this issue, is to actualize the 
slots in the right order, which is made possible by the 
advice. 
Denote by $w_1,w_2,\dots,w_{m'}$ those items from the virtual 
solution $S'$, excluding the 
item remaining in the splitting slot in the end, that are either precious 
or have been in the splitting slot at the end of some phase, 
in the order as they appear during the instance. 
These are exactly the items that the algorithm might miss by not actualizing them. 
We inductively prove that the actual algorithm actualizes the items  
$w_1,w_2,\dots, w_{m'}$.

\subparagraph*{Induction Claim.} 
Immediately after the decision on $w_i$, the algorithm has 
actualized the items 
$w_1,w_2,\dots,w_i$. 

\subparagraph*{Induction Basis.}
Using once more the auxiliary definition of an imaginary item $w_0=u_0=0$ of zero value and 
size, the claim for $i=0$ is empty. 

\subparagraph*{Induction Step.}  
Our induction hypothesis is that the items $w_1,w_2,\dots,w_i$ are already packed and 
actualized immediately after the decision on $w_i$. These items will never be replaced by 
the actual algorithm because the virtual version does not do so either since they are part 
of its eventual solution $S'$. 

The advice tells us into which slot the virtual algorithm will pack 
$w_{i+1}$ in the current 
phase. 
Whenever the virtual version packs an item into this slot, the 
algorithm will try to 
actualize it. 
We only need to prove this to be possible without exceeding the 
given capacity. 

If the oracle advises the algorithm to not actualize any item in the splitting slot in the 
current phase, this is trivial: The actual algorithm never actualizes anything not also 
fully packed by the virtual version, which never exceeds the 
knapsack capacity. 

The more difficult case is that the virtual algorithm puts into the splitting slot during 
this phase some paltry item $i'$ that will be part of the eventual 
virtual 
solution in its entirety rather than just as a split item. 
In the moment when this item is put into the splitting slot, the remaining 
packed paltry items have a total size that will never be 
undercut during the rest of the instance. 
This is because the only way to decrease the total size 
of the packed paltry items is to discard some of them, which would 
be possible only if the 
item of lower priority in the splitting slot were discarded 
beforehand 
by the virtual algorithm, which it cannot do since $i'$ is part of 
the eventual solution. 
This means that $S'$ contains paltry items of some total size at least as large the size of 
paltry items packed when $i'$ appears, plus the entire item $i'$, plus all the items 
$w_1,w_2,\dots,w_{m'}$. Thus the total size of $S'$ is at least what the total size in the 
knapsack becomes for the actual algorithm when packing $i'$ 
fully upon appearance, in addition to those the items 
from $w_1,w_2,\dots,w_{m'}$ that have 
appeared earlier and those yet to appear. 
It might happen that the algorithm has actualized the slot for a 
precious item $w_{j'}$ before it has appeared, which could mean 
that this slot contains an item larger than $w_{j'}$, 
which in turn could block the replacement of the item in the splitting slot by one that has 
higher yield, but is also larger. 
However, recall that in this case the algorithm just re-virtualizes the slots to be 
actualized after the splitting slot in the current phase, dismissing all precious items in 
them. 
Therefore, the algorithm can and will indeed always actualize $w_{i+1}$ upon appearance.

\paragraph*{\textit{Competitive Analysis}}
We analyze the competitive ratios achieved on 
the precious items 
and the 
paltry items separately.
During this, we assume that the algorithm starts in 
the final level and thus never discards previously packed 
items due to leveling 
up; this assumption is justified afterward. 

We still denote by $S$ the optimal solution picked 
by the oracle to 
compute the advice, and let 
$T$ denote the final output of the online 
algorithm. 
\subparagraph*{Precious Part.}
It follows from the termination of all phases in 
the last level and the actualization of all slots 
that the precious part of 
the online solution $T$ contains exactly as many 
precious items in 
each value class as 
$S$ does. Each class spans a value interval whose 
lower and upper 
end spread a factor of exactly 
$1-\eps_\spread$, thus we know that $\val(T\cap 
V_\precious)\ge 
(1-\eps_\spread)\val(S\cap 
V_\precious)$. 
The loss on the precious items is therefore at 
most a fraction 
$\eps_\spread$ of the optimal solution value.  

\subparagraph*{Paltry Part.}
For the paltry items, we have already bounded the 
lost value to a fraction 
$(2\eps_\round+4\eps_\spread/\eps_\purelypaltry)/\eps_\paltry$ 
of the optimal solution when 
discussing how the paltry part of the knapsack 
is packed using the 
splitting slot.

\subparagraph*{Leveling Up.}
It remains to show that the potentially 
unbounded number of resets, 
which occur whenever the algorithm is leveling 
up, are only 
causing a negligible value loss. 

The first trigger causing a level change was 
the appearance of an item 
worth more than the items considered presumably precious so far. 
This trigger ensures that the algorithm never mistreats a precious 
item as paltry since doing so could mean missing a precious item 
that should have filled a slot. If a precious item is not already 
recognized as such upon arrival, the algorithm immediately levels 
up to rectify this and pack the item if necessary.

The second trigger for leveling up was that the 
algorithm can 
construct, from the items seen so far---whether 
they have been accepted or not---a 
valid solution whose value exceeds the 
upper bound $U_\ell$ 
on the optimal solution value $\val(S)$. 
The newest item is part of 
any such hypothetical solution---otherwise the 
level would have 
risen earlier---and when excluding it, the remaining 
possible 
solutions are all 
worth less than the current upper bound $U_\ell$. 

When a level change is occurring---independent of whether it was 
caused by the first or second trigger---this means that the 
estimate 
$U_\ell$ has in fact been far 
too low: 
It has been at most a fraction 
$(1-\eps_\spread)^{K-1}\le\eps_\paltry$ of what it is 
becoming now, as already noted in the section on notions and 
notation. 
Moreover, we have proved, when explaining how the optimal 
solution value $\val(S)$ is estimated, that we can 
derive from any upper 
bound $U_\ell$ a guaranteed lower 
bound 
$L_\ell=U_\ell(1-2\eps_\spread/\eps_\purelypaltry)$.
 
Combining the three facts from the last two paragraphs, 
we see that the 
value of any 
solution that can be constructed from the items 
to be discarded 
with the level change is worth at most 
$\eps_\paltry\val(S)/(1-2\eps_\spread/\eps_\purelypaltry)$.
If the final optimal 
solution contains any of these items, then 
omitting them means 
losing at most a fraction 
$\eps_\paltry/(1-2\eps_\spread/\eps_\purelypaltry)<
2\eps_\paltry$ 
of the overall value, where the last inequality is satisfied for 
our choices of $\eps_\purelypaltry=\eps/2^3$, 
$\eps_\paltry=\eps^2/2^5$, and 
$\eps_\spread=\eps^4/2^{14}$. 
While we might incur 
such a loss 
of a factor bounded by $2\eps_\paltry$ with every new level change, 
the 
aggregate loss of 
the previous resets is fortunately also diminishing in relation to 
the estimate for $\val(S)$ with every update of this estimate to a 
higher value. 
Hence the total fractional 
value loss caused by the potentially unbounded number of resets of 
the algorithm 
can be bounded by 
$\sum_{k=1}^\infty 
(2\eps_\paltry)^k=2\eps_\paltry/(1-2\eps_\paltry)\le
4\eps_\paltry$.
Deducting this fraction from the value of the 
final online 
solution, we are justified in assuming without loss of 
generality that 
the algorithm already starts at the final level. 

\Cref{alg:general_log_bits_upper} implements all of this by 
resetting the 
input sequence to its initial state upon each level change, and 
then 
acting out what it would have done, had it been at the current 
level from the beginning. For all items encountered at past levels 
this is a mere simulation since all previously packed items are 
discarded during a reset. To allow for this simulation, the 
algorithm maintains a set $J_\seen$ of all items seen so far and 
writes them off as lost whenever the level changes. 
The lost items are remembered by storing during each reset the 
current $J_\seen$ as $J_\lost$. This set is then omitted from the 
online output to keep to real. 

\subparagraph*{Conclusion.}
Combining the three types of potential losses 
discussed above---the one on the precious part, the one on the 
paltry part, 
and the one due to leveling up repeatedly---to a 
worst-case scenario, 
we are still left with at least a 
fraction
$1-\eps_\spread-(2\eps_\round+4\eps_\spread/\eps_\purelypaltry)/\eps_\paltry-4\eps_\paltry$
 of the optimal solution value, which is greater than $1-\eps/2\ge 
 1/(1+\eps)$, 
 concluding the proof. 
\end{proof}

\begin{algorithm}[H]
	\caption{{}\textsc{ProPack}}
	\label{alg:general_log_bits_upper}
	\medskip
	\textbf{Parameter:} Any  $\eps\in(0,1/2]$. 

	\medskip
	\textbf{Online Input:} A sequence 
	$I=((s_1,v_1),\dots,(s_n,v_n))$ with the sizes and 
	values of $n$ items.

    \medskip
	\textbf{Online Output:} A packing 
	$T=(T_\paltry\cup\bigcup_{j=0}^m\{\,t_j\mid 
	b_{\actual,j}=1\,\})\setminus J_\lost$ 
	that never exceeds the knapsack 
    capacity and is 
	$(1+\eps)$-competitive in the 
	end. 

	\medskip
	\textbf{Advice:} 
	\begin{itemize}
	\item The modulo class $M(v_1)$ of the first 
	item in the 
	instance, 
	\item the tuple $(M(u_1),\dots,M(u_m))$, where $u_j$ is the 
	$j$th precious item in the
	fixed optimal solution $S$ in order of appearance, and $M(u_j)$ 
	its modulo class, 
	\item thus implicitly $m$,
	\item for every $j\in\{0,\dots,m\}$ the rounded 
	fraction 
	$f_j=\lfloor\val(S\cap 
	V_{\paltry,j}/\val(S)\rfloor_{\eps_\round}$
	 of the optimal solution value 
	$\val(S)$ due to paltry 
	items appearing after $u_j$ but before $u_{j+1}$, 
\item dedicated to one special case, an advice bit\\
	$b_\purelypaltry=
	\begin{cases}1&\text{if greedily packing only paltry items
	is $1/(1-\eps_\purelypaltry)$-competitive}\text{ and}\\
	0&\text{otherwise, and}
	\end{cases}$
\item for every $j\in\{1,\dots,m\}$ a tuple 
$a_j=(a_{j,1},a_{j,2},\dots,a_{j,m_j})$ containing the numbers of 
the slots---with $0$ representing the splitting 
slot---filled during phase $j$ with an item that is entirely part 
of the eventual virtual solution $S'$, in the order in which these 
fillings occur.
	\end{itemize}

	\medskip
	\textbf{Algorithm:}
	\begin{algorithmic}[H]
		\State $\eps_\klein\gets 
		\eps/2^3;\quad\eps_\paltry\gets 
		\eps^2/2^5$ \Comment{Initialize several 
		constants \dots}
		\State $\eps_\spread\gets\eps^4/2^{14};\quad
		\eps_\round\gets\eps^6/2^{20}$ \Comment{\dots{}\;for the given parameter $\eps$.}
		\medskip
		\For{$j\in\{0,1,\dots,m\}$} \Comment{Initialize the 
		slots, where slot $0$ is the splitting slot, as empty \dots}
			\State $t_j\gets 0$ \Comment{\dots{}\;by filling them with an imaginary 
			item $0$ of size and value zero, \dots}
			\State $b_{\actual,j}\gets0$ \Comment{\dots{}\;and 
			initialize the state of the filling to be virtual only.}
		\EndFor
		\State $I_\backup\gets I$ 
		\Comment{Store a copy of the input.}
		\State $T_\paltry \gets \emptyset$ 
		\Comment{Initialize to set of all completely packed paltry 
		items 
		to 
		the empty set.}
		\State $J_\seen \gets \emptyset$; \quad $J_\lost \gets 
		\emptyset$ \Comment{Initialize the 
		sets 
		of seen and lost items to the empty 
		set.}
		\State $p \gets 0$   \Comment{Initialize the 
		phase, 
		equal to the number of 
		virtually filled slots, to zero.}
		\State $q\gets 0$  \Comment{Initialize counter for 
		the number of slot actualized during current phase to $0$.}
		\State $\ell\gets 1$  \Comment{Initialize the 
		level, showing 
		what the algorithm 
		considers precious, to $1$.}
		\medskip	
		\Function{Greedy}{$T',\val_\limit,\size_\limit,i_\split$} 
		\Comment{Greedily reduce the 
		potential packing \dots}
			\State $T' \gets T'\cup\{i_\split\}$ 
			\Comment{\dots{}\;passed to the function, including the split item 
		$i_\split$,  
			\dots}
			\State $T' \gets \{i \in T'\mid \size(i)< 
			\size_\limit\,\}$ 
			\Comment{\dots{}\;after discarding anything of size at 
			least $\size_\limit$, 
			\dots}
			\medskip
			\While{$\size(T')>1$\textbf{ or 
		\	}$\val(T')>\val_\limit$} 
			\Comment{\dots{}\;to a 
			valid solution of value at most 
			$\val_\text{limit}$.}
				\State $T'_\textnormal{LowYield} \gets 
				\argmin\{\,\val(j)/\size(j)\mid 
				j\in T'\,\}$ \Comment{Find 
				items of  
				lowest yield.}
				\State $i_\split \gets 
				\Pop(\argmin\{\,\size(j)\mid 
				j\in 
				T'_\textnormal{LowYield}\,\})$ 
				\Comment{Take a smallest 
				among them.}
				\State $T'\gets 
				T'\setminus\{i_\split\}$  
				\Comment{Separate out from $T'$ the item for the 
				splitting slot.}
			\EndWhile
			\State \Return $(T',i_\split)$ 
			\Comment{Return the 
			reduced solution and 
			the current split item.}
		\EndFunction
        \medskip
		\Function{VirtualSize}{} 
		\Comment{Returns virtual solution size, including used part 
		of split 
		item.}
			\State 
			$\size_\temp\gets\size(T_\paltry\cup\bigcup_{j=1}^m\{\,t_j\})$
			 \Comment{Size of current virtual solution, excluding 
			split item.}
			\State $r\gets\min\{1, (1-\size_\temp)/\size(t_0)\}$ 
			\Comment{Calculate the fraction $r$ of split item to 
			be used.}
			\State \Return 
			$\size(T_\paltry\cup\bigcup_{j=1}^m\{\,t_j\})+r\cdot\size(t_0)$
			\Comment{Return the complete virtual solution.}
		\EndFunction
		\medskip
		\Function{Actual}{} 
		\Comment{Returns online solution of 
		packed paltry items and items in \dots}
			\State \Return 
			$T_\paltry\cup\bigcup_{j=0}^m\{\,t_j\mid 
			b_{\actual,j}=1\,\}$ 
			\Comment{\dots{}\;actualized slots, ignoring level losses.}
		\EndFunction
		\algstore{longalgorithm}
	\end{algorithmic}
\end{algorithm}

\addtocounter{algorithm}{-1}
\begin{algorithm}[H]
	\caption[Test]{{}\textsc{ProPack} (Continuation)}
	\begin{algorithmic}[1]
		\algrestore{longalgorithm}
		\medskip
		\While{$I\neq\emptyset$} \Comment{As long as 
		the instance 
		has 
		not ended, \dots}
			\medskip
			\State $i\gets \Pop(I)$ 
			\Comment{\dots{}\;take 
			the next item in the input sequence.}
			\medskip
			\If $b_\purelypaltry=1$ \Comment{If it is advised to 
			just pack small items of size at most $\eps_\klein$,  
			\dots}
					\State 				
					$(T_\paltry,t_0)\gets\textsc{Greedy}(T_\paltry\cup\{i\},\infty,\eps_\klein,t_0)$
					 \Comment{\dots{}\;then do so greedily, 
					 \dots}
					\State \textbf{continue} \Comment{and continue 
					this 
					way until the instance ends.}
			\EndIf
			\medskip
			\State $J_\seen\gets J_\seen\cup\{i\}$ 
			\Comment{Update 
			the set of 
			items seen so far.}
			\State $b_\levelup\gets 0$ 
			\Comment{Initialize 
			the Boolean checking 
			for a level change to zero.}
			\medskip
			
			\If{$W(i)>\ell$} \Comment{If index of 
			new item's 
			comboclass is greater than current 
			level, \dots}
				\State $\ell\gets W(i)$ \Comment{\dots{}\;then this 
				index becomes the new level.}
				\State $b_\levelup\gets 1$ 
				\Comment{Flag the level change.}
			\EndIf
			\medskip
			\State \LeftComment*{Compute estimated 
			upper bound $U$ 
			for $\val(S)$ based 
			on advice and current level $\ell$:}
			\State $U\gets 
			v_1\sum_{j=1}^m(1-\eps_\spread)^{V(1)+M(u_j)-\ell
K}/(1-\eps_\round(1+/\eps_\paltry)-\sum_{p'=0}^mf_{p'})$
			\medskip
			\State \LeftComment*{Compute optimal 
			solution value 
			attainable with items 
			seen at the current 
			level:}
			\State $\val_\text{max}\gets \max\{\,\val(S') 
			\mid S' 
			\subseteq J, 
			\size(S')\le 1\,\}$ 
			\medskip
			\If{$v_{\max}>U$} \Comment{If this value 
			exceeds the 
			upper bound, \dots}
				\State $\ell\gets \ell+1$ \Comment{\dots{}\;then 
				increase the level by one 
				\dots} 
				\State $b_\levelup\gets 1$ 
				\Comment{\dots{}\;and 
				flag the level 
				change.}
			\EndIf
			\medskip
			\If{$b_\levelup=1$} \Comment{If 
			new item has 
			triggered a level change, then  
			\dots}
				\State $p\gets 
				0$; $q\gets 0$; $T_\paltry\gets\emptyset$; 
				$t_0\gets 0$
				 \Comment{\dots{}\;reset the 
				 algorithm by re-initializing \dots}
				\medskip
				\For{$j\in\{0,1,\dots,m\}$} \Comment{\dots{}\;almost 
				everything, including \dots}
					\State $t_j\gets 0; b_{\actual,j}\gets0$ 
					\Comment{\dots{}\;the slots and their 
					actualization states, \dots}
				\EndFor
				\medskip
				\State $J_\lost\gets J_\seen$; 
				 \Comment{\dots{}\;excepting only the level and 
				 items 
				 seen so far, and,
				 \dots}
				\State 
				$I\gets I_\backup$				
				\Comment{\dots{}\;after restoring the input sequence to the 
				original 
				state, 
				\dots}
				\State \Continue \Comment{\dots{}\;re-start the algorithm, remembering the level and 
				all items lost.}
			\EndIf
			\medskip
			\If $W(i)<\ell$ \Comment{If the new item is 
			provenly paltry, 
			then calculate \dots}
				\State $L\gets 
				U(1-2\eps_\spread/\eps_\purelypaltry)$
				\Comment{\dots{}\;a lower 
				bound on $\val(S)$ close to $U$, and \dots}
				\State $\val_\maxpaltry\gets L\cdot 
				\sum_{p'=0}^pf_{p'}$ 
				\Comment{\dots{}\;use it to 
				derive an 
				upper bound on the 
				\dots}
				\medskip
				\Statex \Comment{\dots{}\;targeted total value 
				of paltry 
				items to be packed before $u_{p+1}$ is presented.}
			\medskip
				\State $t_0'\gets t_0$ \Comment{Store item in 
				splitting slot to 
				see if it will be changed 
				\dots}
				\State $(T_\paltry,t_0)\gets 
				\textsc{Greedy}(T_\paltry\cup\{i\},\val_\maxpaltry,\infty,t_0)$
				 \Comment{\dots{}\;by the update.}
				\medskip
				\If $t_0'\neq t_0$ \textbf{and} $0\in 
				\{a_{p,q},\dots,a_{p,m_p}\}$ 
				\Comment{If so, and splitting slot is to be \dots}
					\medskip
					\While $a_{p,q}\neq0$ \Comment{\dots{}\;actualized or has been, then 
					re-virtualize 
					all \dots}
						\State $b_{\actual,q}\gets0$; $q\gets q-1$ 
						\Comment{\dots{}\;recent slots down to the splitting slot.}
					\EndWhile
				\EndIf
				\medskip
				\If $a_{p,q}=0\textbf{ 
				and }
				\size(\textsc{Actual})+\size(t_0)\le 1$ \Comment{If 
				advised and possible, \dots}
					\State $b_{\actual,0}\gets1$ \Comment{\dots{}\;then (re-)actualize 
					the virtual slot, and \dots}
					\State $q\gets q+1$ \Comment{\dots{}\;increment 
					the counter of 
					slots actualized in the phase.}
				\EndIf
			\EndIf
			\medskip
			\algstore{longalgorithmtwo}
	\end{algorithmic}
\end{algorithm}

\addtocounter{algorithm}{-1}
\begin{algorithm}[H]
	\caption[Test]{{}\textsc{ProPack} (Continuation)}
	\begin{algorithmic}[1]
		\algrestore{longalgorithmtwo}
			\medskip
			\If $W(i)=\ell$ \Comment{If the newly arriving item is 
			presumably precious \dots}
			\medskip
				\If $M(i)=M(u_{p+1})$ 
				\Comment{\dots{}\;and 
				also matches the 
				next empty slot \dots}
					\medskip
					\If $\textsc{VirtualSize}+\size_i\le 1$ 
					\Comment{\dots{}\;without  
					exceeding the knapsack capacity, \dots }
						\State $j\gets p+1$ 
						\Comment{\dots{}\;then 
						store the number of the slot to be filled 
						virtually, 
						\dots}
						\State $p\gets p+1$
						\Comment{\dots{}\;enter the next phase, \dots}
						\State $q\gets 0$
						\Comment{\dots{}\;and 
						accordingly reset the actualization 
						counter.}
					\EndIf 
			\medskip
				\Else \Comment{If new item 
				does not match next 
				empty slot, find among all \dots}
					\State 
					$J_\matches\gets\{\,j'\in\{1,2,\dots,m\}\mid 
					M(t_{j'})=M(i)\,\}$\Comment{\dots{}\;matching slots \dots}
					\State 
					$j\gets\Pop(\arg\max\{\,s(t_{j'})\mid
					j'\in J_\matches\,\})$ \Comment{\dots{}\;one 
					with a 
					largest item.} 
					\State $t_{j}\gets i$
					\Comment{Fill the new item 
					 into the right slot, replacing any present 
					 item.}
				\EndIf
				\medskip
				\If $a_{p,q}=j\textbf{ 
				and }
				\size(\textsc{Actual})+\size(t_j)\le 1$ \Comment{If 
				advised and possible, \dots}
					\State $b_{\actual,0}\gets1$ \Comment{\dots{}\;then actualize 
					the new filling and \dots}					
					\State $q\gets q+1$ \Comment{\dots{}\;update the actualization counter.}
				\EndIf
			\EndIf
			\medskip
		\EndWhile
\State \Return $\textsc{Actual}\setminus J_\lost$ \Comment{Return 
the solution, omitting any items from the lost levels.}
\medskip
	\end{algorithmic}
\end{algorithm}
\bigskip
\bibliography{references}

\end{document}